\documentclass[a4paper,11pt]{article}
\usepackage{xcolor}
\usepackage{amsmath,amsthm}
\usepackage{authblk}
\usepackage{aligned-overset}
\usepackage[style=numeric-comp,maxbibnames=99,bibencoding=utf8,giveninits]{biblatex}
\usepackage{booktabs}
\usepackage{cancel}
\usepackage[hmargin={26mm,26mm},vmargin={30mm,35mm}]{geometry}
\usepackage{nicefrac}
\usepackage{paralist}
\usepackage[colorlinks, allcolors=blue]{hyperref}
\usepackage{subcaption}
\usepackage[compact,small]{titlesec}
\usepackage{tikz}
\usepackage{tikz-cd}
\usepackage{pgfplotstable}
\pgfplotsset{compat=1.18}
\usepackage{newtxmath}
\usepackage{newtxtext}
\usepackage{enumitem}
\usepackage{bm}

\usepackage{graphicx}

\bibliography{hd-adm}

\allowdisplaybreaks

\usepackage[normalem]{ulem}
\newtheorem{theorem}{Theorem}

\newtheorem{lemma}[theorem]{Lemma}

\theoremstyle{remark}
\newtheorem{remark}[theorem]{Remark}
\theoremstyle{definition}
\newtheorem{assumption}[theorem]{Assumption}

\newtheorem{example}[theorem]{Example}

\DeclareRobustCommand{\bvec}[1]{\boldsymbol{#1}}
\newcommand{\ul}[1]{\underline{{#1}}}

\newcommand{\uvec}[1]{\underline{\bvec{#1}}}

\DeclareRobustCommand{\btens}[1]{\boldsymbol{#1}}
\newcommand{\Real}{\mathbb{R}}
\newcommand{\Natural}{\mathbb{N}}

\newcommand{\Symm}{\mathbb{S}}

\newcommand{\Tless}{\mathbb{T}}
\newcommand{\Vspace}{\mathbb{V}}

\newcommand{\st}{\,:\,}

\DeclareMathOperator{\tr}{tr}
\DeclareMathOperator{\Tr}{tr}
\DeclareMathOperator{\DEV}{dev}
\DeclareMathOperator{\SYM}{sym}
\DeclareMathOperator{\skw}{skw}
\DeclareMathOperator{\vskw}{vskw}
\DeclareMathOperator{\mskw}{mskw}

\DeclareMathOperator{\Ker}{Ker}
\DeclareMathOperator{\Image}{Im}
\renewcommand{\Im}{\mathrm{Im}}
\DeclareMathOperator{\ID}{Id}

\newcommand{\DIFF}{\mathrm{d}}

\DeclareMathOperator{\DIV}{div}
\DeclareMathOperator{\VDIV}{\bf div}
\DeclareMathOperator{\DIVDIV}{div \bf div}
\DeclareMathOperator{\CURL}{\bf curl}
\DeclareMathOperator{\GRAD}{\bf grad}
\DeclareMathOperator{\HESS}{\bf hess}
\DeclareMathOperator{\DEF}{\bf def}
\DeclareMathOperator{\INC}{\bf inc}

\DeclareMathOperator{\dev}{dev}
\DeclareMathOperator{\sym}{sym}

\newcommand{\bs}{{\scriptscriptstyle \bullet}}

\newcommand{\Poly}[1]{\mathcal{P}_{#1}}

\newcommand{\uH}[2]{\uvec{X}_{r,#2}^{#1}}

\newcommand{\uI}[2]{\uvec{I}_{r,#2}^{#1}}

\newcommand{\vvvert}{\vert\kern-0.25ex\vert\kern-0.25ex\vert}
\newcommand{\norm}[2]{\Vert #2\Vert_{#1}}
\newcommand{\opn}[2]{\vvvert #2\vvvert_{#1}}

\newcommand{\edges}[1]{\mathcal{E}_{#1}}
\newcommand{\vertices}[1]{\mathcal{V}_{#1}}

\newcommand{\Eh}{\edges{h}}

\newcommand{\Vh}{\vertices{h}}

\title{Hodge-Dirac wave systems and structure-preserving discretizations of the linearized Einstein equations}
\author{Marien-Lorenzo Hanot\thanks{Université de Lille, UMR 8524 - Laboratoire Paul Painlev\'e, CNRS, Inria, France.  Email: \texttt{marien-lorenzo.hanot@univ-lille.fr}},\quad Kaibo Hu\thanks{Mathematical Institute, University of Oxford. Radcliffe Observatory, Andrew Wiles Building,   Oxford OX2 6GG, United~Kingdom. Email: \texttt{kaibo.hu@maths.ox.ac.uk}}}
 
\begin{document}

\maketitle

\tableofcontents

\begin{abstract}
We derive a reformulation of the linearized Arnowitt-Deser-Misner (ADM) equations as a Hodge-Dirac wave system with the divdiv complex, addressing challenges in numerical relativity such as gauge fixing, constraint propagation, and tensor symmetries. The differential and algebraic structures of the divdiv complex ensure the well-posedness of the formulation and facilitate structure-preserving discretization via finite element exterior calculus. 
 We establish the well-posedness of this Hodge-Dirac wave equation and develop a discretization scheme applicable to both conforming and non-conforming discrete complexes, deriving error estimates under minimal assumptions. 
 
\end{abstract}

\section{Introduction}

Numerical relativity is crucial in gravitational wave detection by providing essential templates. A central challenge in this field is solving the Einstein equations numerically. Significant advancements have been made over recent decades, notably the 2005 breakthrough in achieving accurate, long-term evolutions of black hole systems \cite{pretorius2005evolution}. This success highlighted the critical role of the mathematical properties of the Einstein equations, particularly their hyperbolicity, in enabling robust numerical methods.

Despite these achievements, the increasing sophistication of gravitational wave detectors demands greater precision and long-term stability in numerical computations. Moreover, fundamental algorithmic challenges, such as a rigorous understanding of convergence properties, remain unresolved. The absence of comprehensive numerical analytic studies of the Einstein equations also hinders progress in exploring modified gravity models \cite{petrov2023introduction,llinares2018simulation}.

The challenges of numerically solving the Einstein equations stem from at least the following issues:
\begin{enumerate}
\item Gauge fixing and hyperbolicity. As a geometric PDE, the Einstein equations permit reformulations through coordinate choices, known as gauge freedom. The mathematical properties, and thus the numerical performance, of these equations depend critically on the chosen gauge.
\item Constraint propagation and preservation. In 3+1 decompositions, geometric quantities must satisfy constraint equations at each time step, which are inherently preserved by the evolution equations. Maintaining these constraints numerically is both challenging and essential for ensuring numerical stability.
\item Nonlinearity. As characteristic of geometric PDEs, the Einstein equations exhibit strong nonlinearity.
\end{enumerate}
Gauge fixing and constraint preservation necessitate a deep understanding of the differential and algebraic structures of both continuous and discrete equations.
Over recent decades, Finite Element Exterior Calculus (FEEC) \cite{Arnold:18,Arnold.Falk.ea:06,Arnold.Falk.ea:10} has made significant strides in structure-preserving discretization and efficient solvers, with applications in electromagnetism and continuum mechanics.
Differential and cohomological structures have proven critical for accurate and efficient numerical solutions.
In the context of the Einstein equations, FEEC-based approaches were pioneered in \cite{quenneville2015new} using the Einstein-Bianchi (first-order) formulation and in \cite{li2018regge} with Regge finite elements, inspired by Regge calculus.
The Einstein-Bianchi formulation in \cite{quenneville2015new} employs traceless-transverse (TT, i.e., symmetric, traceless, and divergence-free) matrix fields, which naturally arise in conformal complexes \cite{Arnold.Hu:21,vcap2023bgg}.
However, enforcing both symmetry and tracelessness constraints simultaneously poses significant challenges, although some conforming finite element conformal complexes with supersmoothness were recently constructed \cite{hu2023finite,huang2025finite,guo2025Discretizing}.
To encode the tensor symmetries, \cite{quenneville2015new} utilized the Hessian (and $\DIVDIV$) complex with weaker algebraic constraints, imposing them via Lagrange multipliers, following approaches in elasticity \cite{arnold2007mixed}.
Similarly, \cite{Hu.Liang.Ma:21} explored the Einstein-Bianchi system with weakly imposed symmetries.
The work in \cite{li2018regge} extended Regge finite elements and analyzed their properties. We also mention the recent numerical scheme based on a differential form formulation \cite{oliynyk2025polytopal}.
In this paper, we build on the paradigm of structure-preserving discretizations for the Einstein equations. Specifically, we propose a Hilbert complex-based reformulation and its discretization, which naturally incorporates all algebraic and differential constraints {\it strongly} and ensures well-posedness.

We begin with the Arnowitt-Deser-Misner (ADM) equations, which govern the evolution of a metric and its curvature, satisfying the Einstein field equations. These equations arise from a $(3+1)$-decomposition of the field equations, separating temporal and spatial derivatives. A detailed exposition of the $(3+1)$ formalism and the ADM equations can be found, for example, in \cite{Alcubierre:08}. In a $(3+1)$-decomposition, one assumes a foliation of spacetime with spacial slices. The distance between the spacial slices is described by a function $\alpha$, referred to the {\it lapse}; and the {\it shift} between slices is described by a vector quantity $\beta$.

  In this work, we focus on the York version of the ADM equations \cite{York1979}, particularly their linearization, given by
\begin{subequations}  \label{eq:def.York.ADM}
  \begin{align} \label{eq:def.York.ADM1}
    \gamma_{tt} + S \INC \gamma - 2 \HESS \alpha - 2 \DEF \beta_t &= 0, \\ \label{eq:def.York.ADM2}
    \VDIV S (\gamma_t - 2 \DEF \beta) &= 0, \\ \label{eq:def.York.ADM3}
    \DIVDIV S \gamma &= 0,
  \end{align}
\end{subequations}
where $S \colon M_{3\times3}(\Real) \to M_{3\times3}(\Real)$ is an operator defined for any $A \in M_{3\times3}(\Real)$ as $S(A) := A^\top - \Tr(A) I_{3}$, $\gamma \colon \Real^4 \to \Symm$ represents the perturbation of the spatial metric, 
$1+ \alpha \colon \Real^4 \to \Real$ is the lapse function, and $\beta \colon \Real^4 \to \Real^3$ is the shift vector. 
The linearization is performed around $I + \gamma \approx I$ (Euclidean spacial metric), $\beta \approx \bvec{0}$ (zero shift between slices), and $1 + \alpha \approx 1$ (uniform distance between slices). In adapted coordinates, the full spacetime metric $g$ is recovered as:
\[
g :=
\begin{pmatrix}
  -(1+\alpha)^2 + \beta \cdot \beta & \beta^\top \\
  \beta & I + \gamma
\end{pmatrix}.
\]
The lapse function $\alpha$ and shift vector $\beta$ represent coordinate choices rather than physical quantities, serving as gauge functions.
For the operators, $\INC$ is a row-wise curl composed with a column-wise curl;  $\DEF:=\sym\GRAD$ is the symmetric gradient; $ \HESS$ is the Hessian operator of a scalar function; $S\sigma:=\sigma^{T}-\frac{1}{2}\tr(\sigma)I$ is a bijective algebraic operator (see \cite{Arnold.Hu:21}).

The ADM formulation \eqref{eq:def.York.ADM} comprises two sets of equations: \eqref{eq:def.York.ADM1} represents the evolutionary equation, while \eqref{eq:def.York.ADM2}--\eqref{eq:def.York.ADM3} are constraint equations that must be satisfied throughout the evolution. The system exhibits \textit{constraint propagation}, meaning that if \eqref{eq:def.York.ADM2}--\eqref{eq:def.York.ADM3} hold for the initial data, they remain satisfied during the evolution.
The standard ADM formulation \eqref{eq:def.York.ADM} and its nonlinear counterpart lack hyperbolicity. Moreover, straightforward discretization of the evolutionary equations in \eqref{eq:def.York.ADM} introduces numerical errors in the constraints, which accumulate over time and lead to instability.

{\medskip\bf\noindent Hodge-Dirac reformulation.} The first contribution of this paper is to address the loss of constraint propagation, well-posedness and tensor symmetries in numerical discretization by reformulating \eqref{eq:def.York.ADM} as a mixed formulation based on a differential complex.
Then in the framework of finite element exterior calculus  \cite{Arnold.Falk.ea:06,Arnold.Falk.ea:10,Arnold:18}, we can discretize such formulations by discretizing the corresponding complexes and ensure
well-posedness and structure-preserving properties. 

To encode the symmetries of tensors in \eqref{eq:def.York.ADM}, we will reformulate the ADM formulation \eqref{eq:def.York.ADM} into a canonical form based on the $\DIVDIV$ complex \cite{Arnold.Hu:21,vcap2023bgg}:
 \begin{equation} \label{eq:divdiv.complex}
 \begin{tikzcd}[column sep = 1.1cm]
 0  \arrow{r}& \bvec{H}(\DEV\GRAD;\Vspace) \arrow{r}{\DEV\GRAD} &  \bvec{H}(\SYM\CURL;\Tless) \arrow{r}{\SYM\CURL} &  \bvec{H}(\DIVDIV;\Symm) \arrow{r}{\DIVDIV} &  L^2(\Real)\arrow{r}& 0.
 \end{tikzcd}
 \end{equation}
More specifically,
denoting by $\DIFF$ the differential operators in \eqref{eq:divdiv.complex}, i.e., 
$$
\DIFF:=
\begin{pmatrix}
  0 & 0 & 0 & 0 \\
  \DEV\GRAD & 0 & 0 & 0 \\
  0 &\SYM\CURL & 0 & 0 \\
  0 & 0 & \DIVDIV & 0
\end{pmatrix}
$$
with its (formal) adjoint
$$
\DIFF^{\ast}:=
\begin{pmatrix}
  0 & - \VDIV & 0 & 0 \\
 0 & 0 & \CURL & 0 \\
  0 &0 & 0 & \HESS \\
  0 & 0 &0 & 0
\end{pmatrix}
$$
and defining $ {J} := \mathrm{diag}(1, -1, 1, -1)$, we will show that \eqref{eq:def.York.ADM} can be reformulated as
$\partial_t U =  {J} (\DIFF + \DIFF^*) U$ with properly chosen variables $U$
with components from all the spaces in \eqref{eq:divdiv.complex}
(see \eqref{eq:HDADM.HDW}).
This formulation closely resembles a time-dependent, skew-symmetric Hodge-Dirac problem for the de~Rham complex. See \cite{Leopardi.Stern:16} for a study for stationary problems.
 Analogous to how the Hodge-Dirac problem can be viewed as a ``square root'' of a Hodge-Laplace problem, the proposed problem can be regarded as a ``square root'' of the Hodge(-Laplace) wave equation presented in \cite[Section~8.5]{Arnold:18} and \cite{Quenneville:15} with the $\DIVDIV$  complex \eqref{eq:divdiv.complex}.
By analogy, we refer to the problem \eqref{eq:HDADM.HDW} as the {\it Hodge-Dirac wave}.
The primary distinction between the two wave problems is that the Hodge-Laplace equation involves only three consecutive spaces in the complex, whereas the Hodge-Dirac wave considers the entire complex. 
For example, in a complex comprising four spaces, the mixed formulations of the Hodge-Laplace (HL) wave equation and the Hodge-Dirac (HD) wave equation take the forms:
\[
\text{(HL): } \frac{d}{dt} U=
\begin{pmatrix}
  0 & \delta & 0 & 0 \\
  -\DIFF & 0 & -\delta & 0 \\
  0 & \DIFF & 0 & 0 \\
  0 & 0 & 0 & 0
\end{pmatrix}U, \quad
\text{(HD): } \frac{d}{dt}U =
\begin{pmatrix}
  0 & \delta & 0 & 0 \\
  -\DIFF & 0 & -\delta & 0 \\
  0 & \DIFF & 0 & \delta \\
  0 & 0 & -\DIFF & 0
\end{pmatrix}U,
\]
respectively. 
For time-independent equations, stability requires considering the entire complex \cite{Leopardi.Stern:16}. 
The Hodge-Dirac wave equation, and thus the reformulation of the ADM formulation, is well-posed. Each of the four components in the Hodge-Dirac system satisfies a Hodge-Laplacian wave equation; 
while the fourth component in the Hodge-Laplacian wave equation is constant in time, which may exhibit numerical instability under certain perturbations in the nonlinear case. 
Therefore, one may expect that the Hodge-Dirac system has better stability than the Hodge-Laplacian wave.

Compared to previous works that used differential complexes to address systems in numerical relativity, such as \cite{Quenneville:15,Hu.Liang.Ma:21,guo2025Discretizing}, to the best of our knowledge, 
this is the first work that encode all the algebraic symmetries of the tensors and the constraint equations in the ADM-Einstein equations with the $\DIVDIV$ complex. These constraints are thus preserved in numerical discretization.

{\medskip\bf\noindent Discretization of Hodge-Dirac systems.} As another main contribution of this paper, we propose a scheme for solving the Hodge-Dirac wave, establish its well-posedness, and derive error estimates for the time-dependent problem under minimal assumptions. Although presented in three dimensions, the approach is straightforward to extend to any number of dimensions. More importantly, the construction and results do not require conforming discretizations. Any discrete complex may be used, provided it possesses (not necessarily bounded) commuting interpolators and (discrete) Poincaré inequalities. The error estimates then depend on the consistency properties of the chosen spaces. The ability to use non-conforming spaces is particularly valuable for complexes more intricate than the de Rham complex, where conforming discretizations are often challenging to construct and come with other limitations, such as high minimal polynomial degrees or the need for special meshes.
 
The ADM-type formulation discussed in this paper involves the $\DIVDIV$ complex \eqref{eq:divdiv.complex}. The $\DIVDIV$ complex incorporates a second-order differential operator and tensor-valued elements with additional symmetry constraints (e.g., traceless or symmetric elements), posing significant challenges for discretization. We briefly review existing approaches.
\begin{itemize}
\item Some approaches relax strong symmetry constraints and enforce them weakly, as done for the linear elasticity \cite{arnold2007mixed} and the linearized Einstein-Bianchi system \cite{Quenneville:15}.
We choose, however, to preserve strong symmetries within the spaces.
\item Conforming discretizations of the spaces in \eqref{eq:divdiv.complex} were developed in \cite{Chen.Huang:20,Hu.Liang.Ma:21,hu2024family,bonizzoni2025discrete}. Under certain assumptions, conforming spaces and complexes are viable for our scheme. However, most constructions on simplicial meshes are involved due to supersmoothness constraints, which also exclude the use of low-order polynomials. In the numerical results presented in this paper, we use a $\DIVDIV$ discrete complex based on tensor-product splines \cite{bonizzoni2025discrete}. This conforming complex is restricted to Cartesian meshes but is significantly simpler to implement than other alternatives considered. It also yields a highly structured matrix system for numerical schemes. Although the overall degree of local polynomials in the complex can be high (reaching degree $7$ in some spaces), the degree in any Cartesian direction remains below $3$. We discuss this complex in greater detail in Section \ref{sec:num}.
\item To reduce the overall complexity of the method, non-conforming discretizations are also considered. A ``fully discrete'' $\DIVDIV$ complex was developed in \cite{Di-Pietro.Hanot:23*2}, using collections of local polynomials associated with various mesh entities to represent discrete data. This approach supports arbitrary polyhedral elements, not just simplicial meshes, and allows the use of lower-degree polynomials compared to conforming methods. Nevertheless, the spaces remain complex.
\item Another approach involves distributional elements. This concept, explored in \cite{christiansen2011linearization,Braess.Schoberl:08,pechstein2011tangential,gopalakrishnan2020mass,neunteufel2024hellan}, incorporates Dirac deltas into the design of numerical schemes. Distributional spaces can be viewed either through duality with a conforming space or as functions attached to lower-dimensional entities. Such schemes rely on intrinsic finite elements with weaker regularity and their complexes, which have gained attention since Braess and Schöberl's work on a posteriori estimators \cite{Braess.Schoberl:08} and Christiansen's finite element reinterpretation of Regge calculus \cite{christiansen2011linearization}, as well as in recent works \cite{hu2025finite,berchenko2025finite,gopalakrishnan20252,Hu.Lin.Zhang:23,Licht:17}.
In particular, a distributional $\DIVDIV$ complex was introduced in \cite{Hu.Lin.Zhang:23}. While distributional finite elements offer a promising approach, verifying the conditions arising from our analysis remains an open task for future work.
\end{itemize}

The rest of the paper is organized as follows. 
In Section \ref{sec:cont.set},
we reformulate the linearized ADM equations on the $\DIVDIV$ complex.
In Section \ref{sec:cont.wp}, we abstract the resulting formulation to a more general setting 
of ``Hodge-Dirac wave'' and show the well-posedness of the continuous problem.
In Section \ref{sec:var.disc}, we study the discretization of this problem on discrete, 
not necessary conforming, complexes. Obtaining both the well-posedness, and error estimates 
under fairly general assumptions.
Lastly, in Section \ref{sec:num}, 
we introduce some classes of exact solutions, and provide and discuss numerical results 
obtained with the scheme.

\bigskip
{\bf\noindent Notation.} In this paper, we define operators acting on matrix fields {\it column-wise}.  This follows the convention in \cite{Arnold.Hu:21}.  We introduce some notation following  \cite{Arnold.Hu:21}.
\begin{table}[h!]
\begin{center}
\begin{tabular}{c|c}
$\mathbb V$ & $\mathbb R^n$\\
$\mathbb M$ &the space of all $n\times n$-matrices\\
$\mathbb S$ & symmetric matrices\\
$\mathbb K$ & skew symmetric matrices\\
$\mathbb T$ & trace-free matrices\\
$\skw: \mathbb M\to \mathbb K$ & skew symmetric part of a matrix\\
$\sym: \mathbb M\to \mathbb S$ & symmetric part of a matrix\\
$\tr:\mathbb M\to\Real$ & matrix trace\\
$\iota: \Real\to \mathbb M$  & the map $\iota u:= uI$ identifying a scalar with a scalar matrix\\
$\dev:\mathbb{M}\to \mathbb{T}$ & deviator (trace-free part of a matrix) given by $\dev u:=u-1/n \tr (u)I$\\
$S: \mathbb{M}\to \mathbb{M}$ & the map given by $Su = u^T - \tr(u) I$
\end{tabular}
\end{center}
\caption{Notations}
\end{table}
Moreover, $\mskw: \mathbb R^{3}\mapsto \mathbb{K}$ defined by $(\mskw V)_{ij} := - \epsilon_{ijk} V_k$ maps an axial vector to its matrix representation; and  $(\vskw M)_{i} := - \frac12 \epsilon_{ijk} M_{jk}$ takes the skew-symmetric part of a matrix and maps it to the axial vector.

\section{Reformulation of the ADM formulation}
\label{sec:cont.set}

As mentionned in the introduction, we want to find a reformulation of the linearized ADM equations \eqref{eq:def.York.ADM} 
based on the $\DIVDIV$ complex \eqref{eq:divdiv.complex}.
Incorporating the constraint equations of \eqref{eq:def.York.ADM}, 
we propose the following strong formulation:
\begin{subequations}
  \begin{alignat}{2}
    \label{eq:HDADM.0} \VDIV A &= \bvec{0} \\
    \label{eq:HDADM.1} A_t + \CURL S \gamma &= \btens{0} \\
    \label{eq:HDADM.2} S\gamma_t - \SYM\CURL A &= \btens{0} \\
    \label{eq:HDADM.3} \DIVDIV S \gamma &= {0}.
  \end{alignat}
  \label{eq:HDADM}
\end{subequations}
The choice of variables is summarized in the following complex: 
\begin{equation*}
\begin{tikzcd}[column sep=large]
  \Vspace \arrow[r,"\DEV\GRAD"]
  & \Tless \arrow[r,"\SYM\CURL"] \arrow[l,green,shift left=2,"-\VDIV"]
  & \Symm \arrow[r,"\DIVDIV"] \arrow[l,green,shift left=2,"\CURL"]
  & \Real \arrow[l,green,shift left=2,"\HESS"] \\
  {\color{blue} \lambda_0 }& {\color{blue}A} &{\color{blue} S\gamma} &{\color{blue} \lambda_3} 
\end{tikzcd}
\end{equation*}

Let us assume for now that \eqref{eq:HDADM} is well-posed and verify that \eqref{eq:HDADM} is equivalent to the ADM formulation.
\begin{theorem}
  The solution $\gamma$ of \eqref{eq:HDADM} 
    with initial conditions $(A_0,\gamma_0)$
  sastifies the linearized York version of the ADM formulation 
  \eqref{eq:def.York.ADM}
  for the gauge 
  $\alpha := \frac12 \tr \gamma$ and $\beta := \frac12 \int_0^t \VDIV S\gamma + \vskw A_0$.
  \label{th:ADM.equiv}
\end{theorem}
\begin{proof}
  The proof is divided into Lemma \ref{lem:proof.gtt} and Lemma \ref{lem:proof.divgt}.
\end{proof}

\begin{remark}
  The main difficulty is to derive the correct gauge for this formulation. 
  The skew symmetric part of $A$ introduces a non-zero value of $\VDIV S \gamma_t$ 
  which must be corrected using the shift vector $\beta$ (see the proof of Lemma \ref{lem:proof.divgt}).
  While the lapse function $\alpha$ is used to absorb the contribution of $\tr \gamma$ (see the proof of Lemma \ref{lem:proof.gtt}).  
\end{remark}
\begin{remark}
  The main advantage of this system over an $\INC$ based formulation is that,
  without the time derivatives, 
  it reduces to a simple Hodge-Dirac problem.
  The latter is well studied in the context of FEEC (see \cite{Leopardi.Stern:16}), 
  and can be discretized in a stable way using any $\DIVDIV$-complex. 
  Moreover, it strongly enforces both constraints.
  The $\DIVDIV S\gamma = 0$ constraint is directly enforced, 
  while the $\VDIV S\gamma_t = 2\VDIV S\DEF\beta$ can be taken as a definition for the shift vector $\beta$ (see Lemma \ref{lem:proof.gtt}). 
\end{remark}

\begin{lemma}
  If $(A,\gamma)$ is a solution of \eqref{eq:HDADM} then
  \begin{equation}
    \gamma_{tt} + S\INC\gamma - 2\DEF\beta_t - 2\HESS\alpha = \btens{0},
    \label{eq:proof.gtt}
  \end{equation}
  where $\alpha = \frac12 \tr \gamma$ and $\beta = \frac12 \int_0^t \VDIV S\gamma + \vskw A_0$.
  \label{lem:proof.gtt}
\end{lemma}
\begin{proof}
  Taking the time derivative of \eqref{eq:HDADM.2}, we have
  \begin{equation*}
    \begin{aligned}
      \gamma_{tt} - S^{-1}\SYM\CURL A_t &= \bvec{0} \\
      \gamma_{tt} + S^{-1}\SYM\CURL\CURL S\gamma \overset{\eqref{eq:HDADM.1}}&= \bvec{0} \\
      \gamma_{tt} + S^{-1}\left( \INC \gamma - S\DEF\VDIV S\gamma - S\HESS\tr\gamma \right) \overset{\eqref{eq:symcurlcurl.dec}}&= \bvec{0} \\
      \gamma_{tt} + S^{-1}\INC \gamma - 2\DEF\beta_t - 2\HESS\alpha &= \bvec{0} \\
      \gamma_{tt} + S\INC \gamma + \frac12\iota\tr\INC\gamma - 2\DEF\beta_t - 2\HESS\alpha &= \bvec{0} \\
      \gamma_{tt} + S\INC \gamma - \frac12\iota\cancel{\DIVDIV S\gamma} - 2\DEF\beta_t - 2\HESS\alpha \overset{\eqref{eq:HDADM.3}}&= \bvec{0},
    \end{aligned}
  \end{equation*}
  where we used the fact that $S^{-1} = S + \frac12\iota\tr$ on the fifth line, 
  and that $\tr\INC = -\DIVDIV S$ on the last.
\end{proof}

\begin{lemma}
  If $(A,\gamma)$ is solution of \eqref{eq:HDADM} then
  \begin{equation}
    \VDIV S\gamma_t = 2 \VDIV S \DEF\beta,
    \label{eq:proof.divgt}
  \end{equation}
  where $\beta = \frac12 \int_0^t \VDIV S\gamma + \vskw A_0$.
  \label{lem:proof.divgt}
\end{lemma}
\begin{proof}
  Taking the divergence of \eqref{eq:HDADM.2}, we have
  \begin{equation*}
    \begin{aligned}
      \VDIV S\gamma_t 
      &= \VDIV\SYM\CURL A \\
      \overset{\eqref{eq:HDADM.0},\eqref{eq:divsymcurl.dec}}&= -\CURL\CURL\vskw A \\
      &= -\CURL\CURL\vskw \left( A_0 + \int_0^t A_t \right) \\
      \overset{\eqref{eq:HDADM.1}}&= -\CURL\CURL\left( \vskw A_0 - \int_0^t \vskw \CURL S\gamma \right) \\
      &= -\CURL\CURL\left( \vskw A_0 + \frac12\int_0^t \VDIV S^2\gamma \right) \\
      &= -\CURL\CURL\left( \vskw A_0 + \frac12\int_0^t \VDIV S\gamma  - \frac12\int_0^t\GRAD\tr S\gamma\right) \\
      &= -\CURL\CURL\beta + \frac12\int_0^t\cancel{\CURL\CURL\GRAD\tr S\gamma} \\
      \overset{\eqref{eq:curlcurl.divdef}}&= 2\VDIV S\DEF\beta,
    \end{aligned}
  \end{equation*}
  where $A_0$ is the value of $A$ at $t = 0$.
  We used the identity $2\vskw\CURL = -\VDIV S$ on the fifth line, 
  and the fact the $S\gamma \in \Symm$ on the sixth.
\end{proof}

\begin{remark}[Recovering the shift vector]
  The shift vector $\beta$ may be computed  as the solution of 
  \[
    \begin{aligned}
      \CURL\CURL\beta &= -\frac12 \VDIV S\gamma_t \\
      \DIV \beta &= \DIV\vskw A_0.
    \end{aligned}
  \]
\end{remark}

\section{Continuous well-posedness}
\label{sec:cont.wp}
To simplify the notations and emphasize the role played by the divdiv complex \eqref{eq:divdiv.complex}, 
we will denote the differential operators by either $\DIFF$ or $\DIFF^\star$.
In the following, the domain of $\DIFF^\star$ is determined from the domain of $\DIFF$
to make $(\DIFF^\star,D(\DIFF^\star))$ the adjoint of $(\DIFF,D(\DIFF))$.
Using this notation, and writing $g := S\gamma$, the system \eqref{eq:HDADM} becomes:
\begin{equation}
  \begin{aligned}
    \DIFF^\star A &= 0 \\
    A_t + \DIFF^\star g &= 0 \\
    g_t - \DIFF A &= 0 \\
    \DIFF g &= 0.
  \end{aligned}
  \label{eq:HDADM.dver}
\end{equation}

In order to show the well-posedness of \eqref{eq:HDADM.dver}, 
we will first complete it into the following system:
\begin{equation}
  \begin{aligned}
    \partial_t \lambda_0 &= \DIFF^\star A \\
    \partial_t A &= -\DIFF^\star g - \DIFF\lambda_0 \\
    \partial_t g &= \DIFF A + \DIFF^\star\lambda_3 \\
    \partial_t \lambda_3 &= - \DIFF g .
  \end{aligned}
  \label{eq:HDADM.complete}
\end{equation}
We can readily verify that we retrieve \eqref{eq:HDADM.dver} 
from \eqref{eq:HDADM.complete} if $\lambda_0 = 0$ and $\lambda_3 = 0$.

Let 
$$
J := \begin{pmatrix} 1 & 0 & 0 & 0 \\ 0 & -1 & 0 & 0 \\ 0 & 0 & 1 & 0 \\ 0 & 0 & 0 & -1 \end{pmatrix},\quad
\mathcal{L} := J(\DIFF + \DIFF^\star):=\begin{pmatrix} 1 & 0 & 0 & 0 \\ 0 & -1 & 0 & 0 \\ 0 & 0 & 1 & 0 \\ 0 & 0 & 0 & -1 \end{pmatrix}\begin{pmatrix} 0 & \DIFF^\star & 0 & 0 \\ \DIFF  & 0 & \DIFF^\star & 0 \\ 0 & \DIFF & 0 & \DIFF^\star \\ 0 & 0 & \DIFF & 0 \end{pmatrix},   \quad U := \begin{pmatrix} \lambda_0 \\ A \\ g \\ \lambda_3 \end{pmatrix}.
$$
The system \eqref{eq:HDADM.complete} can be written in a more condensed form as 
\begin{equation}
  \partial_t U =  \mathcal{L} U .
  \label{eq:HDADM.HDW}
\end{equation}
We define the domain of $\mathcal{L}$ as
\[
  D(\mathcal{L}) := D(\DIFF) \cap D(\DIFF^\star).
\]

\begin{lemma}[Skew adjointness] \label{lem:skew.adj}
  The operator $\mathcal{L}$ is skew-adjoint as an unbounded operator.
\end{lemma}
\begin{proof}
  Noticing that, 
  for all $0 \leq k < 3$, 
  $v\in D(\DIFF)$ such that only the $k$-th component of $v$ is non-zero,
  $\DIFF J v = (-1)^{k+1} \DIFF v$, while $J\DIFF v = (-1)^{k} \DIFF v$, 
  we have $\DIFF J = - J \DIFF$.
  Therefore, since $J$ is self-adjoint,
  we have formally
  \[
    (J\DIFF)^\star = \DIFF^\star J^\star = \DIFF^\star J = - J \DIFF^\star.
  \]
  Using the same identity for $J\DIFF^\star$, we find:
  \[
    \mathcal{L}^\star = (J\DIFF + J\DIFF^\star)^\star = -(J\DIFF^\star + J\DIFF) = -\mathcal{L}.
  \]
  Viewing $\mathcal{L}$ as an unbounded operator defined on $D(\DIFF)\cap D(\DIFF^*)$, 
  we find the domain of $\mathcal{L}^\star$ to be $D(\DIFF^\star)\cap D(\DIFF) = D(\mathcal{L})$.
  Therefore $(\mathcal{L},D(\mathcal{L}))$ is skew-adjoint as an unbounded operator.
\end{proof}

The main tool to study the system \eqref{eq:HDADM.HDW} is the following result from 
\cite[Proposition~4.1.6,Corollary~2.4.9]{Cazenave.Haraux:98}:
\begin{theorem}[Hille-Yosida] \label{thm:Hille.Yosida}
Given a Hilbert space $X$, a skew-adjoint operator $A$ with domain $D(A)$ dense in $X$,
$U_0 \in D(A)$ and $f \in C([0,T],X)$ such that $f \in W^{1,1}([0,T],X)$ or $f \in L^1( (0,T),D(A))$.
There exists a unique solution 
\[
  U \in C^0([0,T],D(A)) \cap C^1([0,T],X)
\]
to the system $\partial_t{U} = AU + f$ with $U(0) = U_0$.
Moreover, if $f \in L^1( (0,T),X)$ then by \cite[Lemma~4.1.5]{Cazenave.Haraux:98},
\[
  \norm{C([0,T],X)}{U} \leq \norm{X}{U_0} + \norm{L^1((0,T),X)}{f}.
\]
\end{theorem}
\begin{remark}
From \cite[Proposition 6.1.1]{Cazenave.Haraux:98}, we also have the energy estimate for skew-symmetric operators
$$
\frac12 \frac{\DIFF}{\DIFF t} \Vert U \Vert^2 = <f,U>.
$$
\end{remark}

Let $X = L^2(\Omega,\Vspace) \oplus 
  L^2(\Omega,\Tless) \oplus 
    L^2(\Omega,\Symm) \oplus 
      L^2(\Omega,\Real)$. 
\begin{theorem}[Well-posedness]
  If $D(\mathcal{L})$ is dense in $X$, 
  and the initial condition $(\lambda_0^0,A_0,S\gamma_0,\lambda_3^0) \in D(\mathcal{L})$, 
  then there exists a unique solution $(\lambda_0,A,S\gamma,\lambda_3)$ to the system \eqref{eq:HDADM.complete}.
  Moreover, its norm is bounded by
  \[
    \norm{C([0,T],X)}{(\lambda_0,A,S\gamma,\lambda_3)} \leq \norm{X}{(\lambda_0^0,A_0,S\gamma_0,\lambda_3^0)}.
  \]
  If the initial conditions are compatible with \eqref{eq:HDADM} 
  (i.e. $\lambda_0^0 = 0$, $\lambda_3^0 = 0$, $\VDIV A_0 = \bvec{0}$, and $\DIVDIV S\gamma_0 = 0$),
  then $(A,S\gamma)$ is the unique solution of the system \eqref{eq:HDADM}.
  \label{thm:cont.wp}
\end{theorem}
\begin{proof}
  If the initial condition $(\lambda_0^0,A_0,S\gamma_0,\lambda_3^0) \in D(\mathcal{L})$,
  then Lemma \ref{lem:skew.adj} allows us to apply Theorem \ref{thm:Hille.Yosida} 
  with $f = 0$ to show the existence and uniqueness of a solution $U := (\lambda_0,A,S\gamma,\lambda_3)$
  of \eqref{eq:HDADM.HDW}, together with the bound on the norm.

  Let us now assume that $\lambda_0^0 = 0$, $\lambda_3^0 = 0$, $\VDIV A_0 = \bvec{0}$, and $\DIVDIV S\gamma_0 = 0$.
  It remains to prove that $\lambda_0 = 0$ and $\lambda_3 = 0$ to retrieve \eqref{eq:HDADM.dver}. 
  To this end, we will first show that $W := (0,A - P_{\Im \DEV\GRAD} A,S\gamma - P_{\Im \HESS} S\gamma,0)$
  is another solution and concludes with the uniqueness of the solution.
  Using the Hodge decomposition
  \[
    L^2 \otimes \Tless = \Im \DEV\GRAD \oplus \Im \CURL \oplus \mathfrak{H}^1,\quad
    L^2 \otimes \Symm = \Im \HESS \oplus \Im \SYM\CURL \oplus \mathfrak{H}^2,
  \]
where $\mathfrak{H}^1$ and $\mathfrak{H}^2$ are the (possibly empty) spaces of harmonic forms,
  we deduce that $\VDIV(A - P_{\Im \DEV\GRAD} A) = \VDIV(P_{\Im\CURL \oplus \mathfrak{H}^1} A) = 0$, 
  and $\DIVDIV(S\gamma - P_{\Im \HESS} S\gamma) = \DIVDIV P_{\Im \SYM\CURL \oplus \mathfrak{H}^2} S\gamma = 0$.
  Moreover, we have $\SYM\CURL(A - P_{\Im \DEV\GRAD} A) = \SYM\CURL A$, 
  and $\CURL(S\gamma - P_{\Im \HESS} S\gamma) = \CURL S\gamma$.
  Therefore, at all time $t$, $W \in D(\mathcal{L})$. 
  Since the orthogonal projections are continuous and commute with the time derivative, 
  we infer that $W \in C^0([0,T],D(\mathcal{L})) \cap C^1([0,T],X)$.
  Since $\VDIV A_0 = 0$, we have $(A - P_{\Im \DEV\GRAD} A)_0 = A_0$, 
  and, likewise $\DIVDIV S\gamma_0 = 0$ implies $(S\gamma - P_{\Im \HESS} S\gamma)_0 = S\gamma_0$. 
  It only remains to check that $\partial_t{W} = \mathcal{L}W$,
  that we infer from the Hodge decomposition 
  giving $\partial_t A = \partial_t P_{\Im \DEV\GRAD} A - \CURL S \gamma$, 
  and $\partial_t S\gamma = \partial_t P_{\Im \HESS}S\gamma + \SYM\CURL A$.

  Therefore $(0,A - P_{\Im \DEV\GRAD} A,S\gamma - P_{\Im \HESS} S\gamma,0)$ is a solution of \eqref{eq:HDADM.HDW}, 
  and by uniqueness we must have $U = W$, hence $\lambda_0 = 0$ and $\lambda_3 = 0$.
  Thus $(A,\gamma)$ is a solution of \eqref{eq:HDADM}. 
  We readily verify that any solution of \eqref{eq:HDADM} gives a solution of \eqref{eq:HDADM.HDW}, 
  showing the uniqueness of $(A,\gamma)$.
\end{proof}

\section{Variational formulation and discretization}
\label{sec:var.disc}

Consider the problem: 
Given $U_0 \in D(\DIFF) \cap D(\DIFF^\star)$, 
find $U \in C^1([0,T],X) \cap C^0([0,T],D(\DIFF))$,
such that $U(0) = U_0$, and for all $t\in[0,T]$, all $V\in D(\DIFF)$, 
\begin{equation}
  (\partial_t U, V) = (J\DIFF U,V) - (U,J\DIFF V).
  \label{eq:HDADM.var}
\end{equation}
\begin{lemma}
  A function $U$ is a solution of \eqref{eq:HDADM.var} if and only if $U$ is a solution of \eqref{eq:HDADM.HDW}.
\end{lemma}
The well-posedness of \eqref{eq:HDADM.var} follows from Theorem \ref{thm:cont.wp}.

Let us now consider the discretization of this problem.
For the discretization in space, we consider a discrete complex:
\begin{equation*}
  \begin{tikzcd}
    L^2(\Omega,\Vspace) \arrow[r,"\DIFF"] \arrow[d,"\uI{0}{h}"]
    & L^2(\Omega,\Tless) \arrow[r,"\DIFF"] \arrow[d,"\uI{1}{h}"]
    & L^2(\Omega,\Symm) \arrow[r,"\DIFF"] \arrow[d,"\uI{2}{h}"]
    & L^2(\Omega,\Real) \arrow[d,"\uI{3}{h}"] \\
    \uH{0}{h} \arrow[r,"\DIFF_h"]
    & \uH{1}{h} \arrow[r,"\DIFF_h"]
    & \uH{2}{h} \arrow[r,"\DIFF_h"]
    & \uH{3}{h}.
  \end{tikzcd}
\end{equation*}
We denote by $\uH{}{h} := \bigtimes_{i=0}^3 \uH{i}{h}$ the global discrete space,
and introduce two norms on the discrete space: 
the $L^2$-like norm $\norm{h}{V} := \sqrt{\sum_{i=0}^3(v_i,v_i)_h}$, 
and the graph norm $\norm{1,h}{V} := \norm{h}{V} + \norm{h}{\DIFF_h V}$.
\begin{assumption}
  The discrete complex must satisfy the following properties:
  \begin{enumerate}[label=\textbf{(A{\arabic*})}, series=asm_enum]
    \item The discrete complex admits uniform Poincaré inequalities:
      there is $c_p$ independant of the mesh size $h$, such that for all $0 \leq i < 3$,
      \[
        \forall \uvec{\tau}_h \in \uH{i}{h} \cap (\Ker \DIFF_h)^{\perp_h},
        \norm{h}{\uvec{\tau}_h} \leq c_p \norm{h}{\DIFF_h \uvec{\tau}_h},
      \]
      where $ (\Ker \DIFF_h)^{\perp_h}$ denotes the orthogonal complement of $\Ker \DIFF_h$ with respect to the $\norm{h}{\bs}$ inner product.
      \label{asm:disc.spaces.poincare}
    \item The interpolator is a cochain map, i.e. $\forall 0 \leq i < 3$, $\DIFF_h\uI{i}{h} = \uI{i+1}{h}\DIFF$.
      \label{asm:disc.spaces.cochain}
  \end{enumerate}
  \label{asm:disc.spaces}
\end{assumption}
Notice that the right-hand side of \eqref{eq:HDADM.var} is, in general, degenerate.
If $W$ is such that $\DIFF W = 0$ and $\DIFF^* W = 0$ (if $W$ is an harmonic form), then 
\eqref{eq:HDADM.var} becomes $\partial_t W = 0$. 
Since the equation is linear, this simply means that the harmonic component of a solution is preserved by the evolution, 
and it does not require a particular treatment in the implementation.
We introduce the notation $P_{\mathfrak{H}}$ for the $L^2$ orthonormal projector into the subspace of harmonic forms.
We use the same notation for the projector acting on the continuous and the discrete spaces; 
its meaning can be deduced from its argument.
With this convention of notation, we have on both continuous and discrete levels
\begin{equation}
  \partial_t P_{\mathfrak{H}} U = 0.
  \label{eq:Phdt}
\end{equation}
We denote by $\mathfrak{H}_h \subset \uH{}{h}$ the space of discrete harmonic forms.

\subsection{Spatial discretization}
We first consider the discretization of the spatial derivative appearing in \eqref{eq:HDADM.var}. 
We define the bilinear operator $\mathcal{L}_h : \uH{}{h}\times\uH{}{h}\to\Real$ 
for all $V_h,W_h\in\uH{}{h}$ by
\begin{equation}
  \mathcal{L}_h(V_h,W_h) := (J\DIFF_h V_h,W_h)_h - (V_h,J\DIFF_h W_h)_h.
  \label{eq:def.Lh}
\end{equation}
Notice that $\mathcal{L}_h$ is continuous for the graph norm.

\begin{lemma}[Partial Inf-Sup stability of $\mathcal{L}_h$]
  The bilinear form $\mathcal{L}_h$ is inf-sup stable on $\mathfrak{H}^{\perp_h}$:
  there is $C_L > 0$ depending only on $c_p$ such that, 
  for all $V_h \in \mathfrak{H}^{\perp_h}$, 
  \[
    \sup_{W_h\in\mathfrak{H}^{\perp_h}, W_h\neq 0} \frac{\mathcal{L}_h(V_h,W_h)}{\norm{1,h}{W_h}} \geq C_L \norm{1,h}{V_h}.
  \]
  \label{lem:infsup.Lh}
\end{lemma}
\begin{proof}
  Let $V_h\in\mathfrak{H}^{\perp_h}$. 
  We infer from the discrete Poincaré inequality Assumption \ref{asm:disc.spaces.poincare}
  the existence of $\rho_h\in\uH{}{h}\cap(\Ker\DIFF_h)^{\perp_h}$ such that $\DIFF_h\rho_h = -JP_{\Image \DIFF_h} V_h$ and $\norm{h}{\rho_h} \leq c_p \norm{h}{V_h}$.
  Since $V_h\in\mathfrak{H}^{\perp_h}$, we have $P_{\Image \DIFF_h} V_h = P_{\Ker \DIFF_h} V_h$.
  Setting $W_h = J\DIFF_h V_h + \rho_h$, we have
  \begin{equation}
    \begin{aligned}
      \mathcal{L}_h(V_h,W_h) 
      &:= (J\DIFF_h V_h, J\DIFF_h V_h)_h + (J\DIFF_h V_h, \rho_h)_h
      - (V_h,J\DIFF_h\rho_h)_h \\
      & = \norm{h}{\DIFF_h V_h}^2 + 0 + \norm{h}{P_{\Ker\DIFF_h}V_h}^2 \\
      & \geq \frac12 \norm{h}{\DIFF_h V_h}^2 + \frac{1}{2 c_p^2} \norm{h}{(\ID - P_{\Ker\DIFF_h})V_h}^2 
      + \norm{h}{P_{\Ker\DIFF_h}V_h}^2 \\
      & \geq \frac{1}{2 \max\lbrace 1, c_p^2 \rbrace} \norm{1,h}{V_h}^2 .
    \end{aligned}
    \label{eq:infsup.Lh.P1}
  \end{equation}
  Moreover, from the construction of $W_h$, we have
  \begin{equation}
    \norm{1,h}{W_h} = \norm{h}{\DIFF_h V_h} + \norm{h}{\rho_h} + \norm{h}{\DIFF_h \rho_h} 
    \leq (1+c_p) \norm{h}{U} + \norm{h}{\DIFF_h V_h} \leq (1+c_p)\norm{1,h}{V_h} .
    \label{eq:infsup.Lh.P2}
  \end{equation}
  We infer the result dividing \eqref{eq:infsup.Lh.P1} by \eqref{eq:infsup.Lh.P2}.
\end{proof}

\subsection{Time discretization}
For the discretization in time, we consider a spliting of the time interval 
$t_0 = 0 < t_1 < \dots < t_N = T$ and, for each space of the complex \eqref{eq:divdiv.complex},
an linear operator approximating the time derivative
at $t_n$, 
$\mathcal{T}^n := \mathcal{T}_B^n + \mathcal{T}_L^n$, together with its discrete counterpart 
$\mathcal{T}_h^n := \mathcal{T}_{h,B}^n + \mathcal{T}_{h,L}^n$.
The spliting is done to seperate the implicit part $\mathcal{T}_B^n$ from the explicit part $\mathcal{T}_L^n$ of the operator. 
Explicitly, the component 
      $\mathcal{T}^n_{h,L}(V_h^m)_{0 \leq m \leq n} = \mathcal{T}^n_{h,L} (V_h^m)_{0 \leq m < n)}$
      is a vector of $\uH{i}{h}$ that does not depends on $V_h^n$, 
      while the component $\mathcal{T}^n_{h,B}(V_h^m)_{0 \leq m \leq n} = \mathcal{T}^n_{h,B}V_h^n$
      is linear in $V_h^n$. 
\begin{example}
  For the Euler method, the time discretization operator reads:
  \[
    \mathcal{T}_{h,B}^n((V^m_h)_{m\leq n}) := \frac{V^n}{t_n-t_{n-1}}, \quad 
    \mathcal{T}_{h,L}^n((V^m_h)_{m\leq n}) := \frac{-V^{n-1}}{t_n-t_{n-1}}.
  \]
\end{example}
Moreover, the time discretization must satisfy the following properties:
\begin{assumption}
  For all $0 \leq n \leq N$, and all $0 \leq i \leq 3$,
  \begin{enumerate}[resume*=asm_enum]
    \item The operators commute with the interpolator, i.e. 
      $\uI{i}{h}\mathcal{T}^n_B = \mathcal{T}^n_{h,B} \uI{i}{h}$
      and $\uI{i}{h}\mathcal{T}^n_L = \mathcal{T}^n_{h,L} \uI{i}{h}$.
      \label{asm:disc.time.cochain}
    \item The operator commutes with the spatial derivatives:
      $\DIFF\mathcal{T}^n = \mathcal{T}^n\DIFF$ and $\DIFF^\star\mathcal{T}^n = \mathcal{T}^n\DIFF^\star$.
      \label{asm:disc.time.com}
    \item There exists $\theta_n > 0$, such that for all $V_h\in\uH{i}{h}$, 
      $(\mathcal{T}^n_{h,B}V_h,V_h)_h \geq \theta_n \norm{h}{V_h}^2$.
      Moreover, we must have $\sum_{n=1}^N \frac{1}{\theta_n} \leq C T$ 
      for some constant $C$ independant of $N$. 
      \label{asm:disc.time.coer}
    \item The discrete operator commutes with $P_{\mathfrak{H}}$: $P_{\mathfrak{H}}\mathcal{T}^n_h = \mathcal{T}^n_h P_{\mathfrak{H}}$, 
      and the continuous operator is zeroth order consistent: 
      if $W$ is constant in time then $\mathcal{T}^n W = 0$.
      \label{asm:disc.time.harmonics}
  \end{enumerate}
  \label{asm:disc.time}
\end{assumption}
The time discretization operators are defined on each space $\uH{i}{h}$ and extend to $\uH{}{h}$ diagonally.
In order to quantify the error on the discrete solution,
we introduce the following notations:
\begin{align}
  \epsilon_{\Delta t,n} (V) &:= \mathcal{T}^n V - \partial_t V(t_n), \label{eq:def.espT} \\
  \epsilon_h (V,W_h) &:= (\uI{}{h}V,\DIFF_h W_h)_h - (\uI{}{h}\DIFF^\star V,W_h)_h. \label{eq:def.esph}
\end{align}

The discrete problem is: 
Given $U_h^0 \in \uH{}{h}$, 
find $(U_h^n)_{1\leq n\leq N} \in \left(\uH{}{h}\right)^N$ 
such that for all $1 \leq n \leq N$, 
and all $V_h \in \uH{}{h}$, 
\begin{equation}
  a_h^n(U_h^n,V_h) = l_h^n(V_h),
  \label{eq:disc.var}
\end{equation}
where,  
\begin{equation*}
  a_h^n(U_h^n,V_h) := (\mathcal{T}^n_{h,B}U_h^n,V_h) - \mathcal{L}_h(U_h^n,V_h),
  \quad l_h^n(V_h) := -(\mathcal{T}^n_{h,L}(U_h^m),V_h)_h.
\end{equation*}

\begin{theorem}[Error estimate]
  There exists a unique solution $(U_h^n)$ to the problem \eqref{eq:disc.var}.
  Moreover, denoting by $U$ the solution of the continuous problem \eqref{eq:HDADM.var}, 
  if for all $1 \leq n \leq N$, $(V_h^n)_{n\leq N}\in(\uH{}{h})^N$,
  \begin{equation*}
    \begin{aligned}
      \norm{h}{\mathcal{T}^n_{h,L} (V_h^m)_{m<n} } &\leq \theta_n \norm{h}{V_h^{n-1}}, \\
      \norm{h}{\uI{}{h} \epsilon_{\Delta t,N}(U)} + \norm{h}{\uI{}{h} \mathcal{T}^n (\epsilon_{\Delta t,m}(U))_{m\leq n}} &\leq (\Delta t)^l E(\vert U \vert), \\
      \opn{1,h}{\epsilon_h(U,\cdot)} + \opn{1,h}{\epsilon_h(\mathcal{T}^nU,\cdot)} &\leq h^{r+1} E(\vert U \vert),
    \end{aligned}
  \end{equation*}
  with $r,l \in \Natural$ the order of convergence in space and time, 
  $h$ the characteristic size of the discrete spaces, 
  and $E(\vert U \vert)$ proportional to some semi-norm of $U$, 
  then, 
  \[
    \norm{h}{U_h^N - \uI{}{h}U(T)} \leq ((\Delta t)^l + h^{r+1}) C E(\vert U \vert) + 
    \norm{h}{(P_{\mathfrak{H}}\uI{}{h} - \uI{}{h}P_{\mathfrak{H}})U},
  \]
  for some constant $C$ independant of $\Delta t$ and $h$.
  \label{thm:err.estimate}
\end{theorem}
\begin{proof}
  The proof is detailed in Section \ref{sec:err.estimate.proof}
\end{proof}
\begin{remark}
  The last term can be seen as the harmonic gap. In many cases, it can be estimated from other argument.
  For instance, if the domain is contractible, the only harmonic forms are in the space of $0$-forms (or $3$-forms if enforcing Dirichlet boundary conditions).
  Since the $0$ and $3$-forms components of $U$ are zero, we have 
  $\norm{h}{(P_{\mathfrak{H}}\uI{}{h} - \uI{}{h}P_{\mathfrak{H}})U} = 0$.
\end{remark}
\begin{remark}
  In order to simplify the proof, we restrict ourselves to first order time discretizations.
  However, only a simple adaptation is necessary to consider more general time discretizations.
\end{remark}

\subsection{Proof of the error estimate} \label{sec:err.estimate.proof}

\begin{lemma}[Coercivity for the $L^2$-norm]
  \label{lem:disc.infsup}
  The bilinear form $a_h^n$ is coercive for the norm $\norm{h}{\cdot}$:
  \[
    \forall V_h \in \uH{}{h}, 
    \quad a_h^n(V_h,V_h) \geq \theta_n \norm{h}{V_h}^2 .
  \]
\end{lemma}
\begin{proof}
  The result stems from the skew-symmetry of the operator $\mathcal{L}$, and of the coercivity \ref{asm:disc.time.coer} of $\mathcal{T}_{h,B}^n$.
  Evaluating $a_h^n(V_h,V_h)$, we have:
  \begin{equation*}
    \begin{aligned}
      a_h^n(V_h,V_h) &= (\mathcal{T}_{h,B}^nV_h,V_h)_h 
      - (J\DIFF_h V_h,V_h)_h + (V_h,J\DIFF_h V_h)_h \\
      &\geq \theta_n \norm{h}{V_h}^2.
    \end{aligned}
  \end{equation*}
\end{proof}

Let $U$ be the solution of \eqref{eq:HDADM.var}.
In order to derive the error estimate, we first introduce another problem:
For all $0 < n \leq N$, find $\hat{U}_h^n\in\uH{}{h}$, such that for all $V_h\in\uH{}{h}$, 
\begin{equation}
  \mathcal{L}_h((\ID - P_{\mathfrak{H}})\hat{U}_h^n,V_h) = ((\ID - P_{\mathfrak{H}})\uI{}{h}\mathcal{T}^nU,V_h)_h, \quad P_{\mathfrak{H}} \hat{U}_h^n = \uI{}{h}P_{\mathfrak{H}}U(t_n).
  \label{eq:def.hatUh}
\end{equation}
\begin{remark}
  The difficulty here is that we must use the graph norm to get the error estimates on the spatial discretization, 
  but we also need to sharply control the $L^2$ norm to preserve the convergence in time.
  Hence the need of this auxiliary problem.
\end{remark}

\begin{lemma}[Error estimate from previous steps]
  \label{lem:err.time}
  Let $(U_h^n)_n$ be solution of \eqref{eq:disc.var}, 
  and $(\hat{U}_h^n)_n$ be solution of \eqref{eq:def.hatUh}.
  For any $0 < n \leq N$, it holds:
  \[
    \norm{h}{\hat{U}_h^n - U_h^n} \leq 
  \frac{1}{\theta_n}\norm{h}{(\ID-P_{\mathfrak{H}})\left( \mathcal{T}^n_h(\hat{U}_h^m)_{m\leq n} -  \uI{}{h}\mathcal{T}^nU\right)}
    + \frac{1}{\theta_n}\norm{h}{\mathcal{T}^n_{h,L}\left( (\hat{U}_h^m)_{m<n} - (U_h^m)_{m<n} \right)}.
  \]
\end{lemma}
\begin{proof}
  Inserting the definition \eqref{eq:def.hatUh} of $\hat{U}^n_h$ into the definition of $a_h^n$, we infer
  that for all $V_h\in\uH{}{h}$,
  \[
    a_h^n(\hat{U}_h^n,V_h) = (\mathcal{T}^n_{h,B}\hat{U}_h^n,V_h)_h - ( (\ID-P_{\mathfrak{H}})\uI{}{h}\mathcal{T}^nU,V_h)_h.
  \]
  Moreover, we infer from Assumption \ref{asm:disc.time.harmonics}, \eqref{eq:Phdt} and \eqref{eq:def.hatUh} that 
  \[
    P_{\mathfrak{H}} \mathcal{T}^n_h(\hat{U}_h^m)_{m\leq n} =
    \mathcal{T}^n_h(P_{\mathfrak{H}}\hat{U}_h^m)_{m\leq n} =
    \mathcal{T}^n_h(\uI{}{h}P_{\mathfrak{H}}U) =
    \uI{}{h}\mathcal{T}^n(P_{\mathfrak{H}}U) = 0.
  \]
  Therefore, from the definition \eqref{eq:disc.var} of $U_h^n$, we have 
  \begin{equation}
    \begin{aligned}
      a_h^n(\hat{U}_h^n - U_h^n,V_h) 
      &= (\mathcal{T}^n_{h,B}\hat{U}_h^n -  (\ID-P_{\mathfrak{H}})\uI{}{h}\mathcal{T}^nU + \mathcal{T}^n_{h,L}(U_h^m)_{m<n} ,V_h)_h\\
      &= (\mathcal{T}^n_{h,B}\hat{U}_h^n + \mathcal{T}^n_{h,L}(\hat{U}_h^m)_{m<n} -  (\ID-P_{\mathfrak{H}})\uI{}{h}\mathcal{T}^nU,V_h)_h 
      + (\mathcal{T}^n_{h,L}(U_h^m)_{m<n} - \mathcal{T}^n_{h,L}(\hat{U}_h^m)_{m<n},V_h)_h \\
      &= (\mathcal{T}^n_h(\hat{U}_h^m)_{m\leq n} -  (\ID-P_{\mathfrak{H}})\uI{}{h}\mathcal{T}^nU ,V_h)_h 
      + (\mathcal{T}^n_{h,L}\left( (\hat{U}_h^m)_{m<n} - (U_h^m)_{m<n} \right),V_h)_h\\
      &= ((\ID-P_{\mathfrak{H}})\left( \mathcal{T}^n_h(\hat{U}_h^m)_{m\leq n} -  \uI{}{h}\mathcal{T}^nU\right),V_h)_h 
      + (\mathcal{T}^n_{h,L}\left( (\hat{U}_h^m)_{m<n} - (U_h^m)_{m<n} \right),V_h)_h.
    \end{aligned}
    \label{eq:err.time.P1}
  \end{equation}
  We conclude from Lemma \ref{lem:disc.infsup} and \eqref{eq:err.time.P1}, writting 
  \begin{equation*}
    \begin{aligned}
    \norm{h}{\hat{U}_h^n - U_h^n}^2 &\leq \frac{1}{\theta_n} a_h^n(\hat{U}_h^n - U_h^n,\hat{U}_h^n - U_h^n) \\
    &\leq \frac{1}{\theta_n} \left( \norm{h}{(\ID-P_{\mathfrak{H}})\left( \mathcal{T}^n_h(\hat{U}_h^m)_{m\leq n} -  \uI{}{h}\mathcal{T}^nU\right)}
    + \norm{h}{\mathcal{T}^n_{h,L}\left( (\hat{U}_h^m)_{m<n} - (U_h^m)_{m<n} \right)} \right) \norm{h}{\hat{U}_h^n - U_h^n} .
    \end{aligned}
  \end{equation*}
\end{proof}

\begin{lemma}[Error estimate at a single step]
  Let $(U_h^n)_n$ be solution of \eqref{eq:disc.var}, 
  and $(\hat{U}_h^n)_n$ be solution of \eqref{eq:def.hatUh}.
  For any $0 < n \leq N$, it holds:
  \[
    \norm{1,h}{(\ID-P_{\mathfrak{H}})\left(\uI{}{h}\mathcal{T}^nU - \mathcal{T}^n_h (\hat{U}_h^m)_{m\leq n}\right)} \leq \frac{1}{C_L}\left( \opn{1,h}{\epsilon_h(\mathcal{T}^nU,\cdot)} + \norm{h}{\uI{}{h} \mathcal{T}^n (\epsilon_{\Delta t,m}(U))_{m\leq n}}\right),
  \]
  where $C_L$ is the constant given by Lemma \ref{lem:infsup.Lh}.
  \label{lem:err.space}
\end{lemma}
\begin{proof}
  Applying the linearity of the time discretization operator $\mathcal{T}^n_h$ to \eqref{eq:def.hatUh}, we have
  for any $V_h \in \mathfrak{H}^{\perp_h}$:
  \begin{equation}
    \mathcal{L}_h((\ID-P_{\mathfrak{H}})\mathcal{T}^n_h (\hat{U}_h^m)_{m\leq n},V_h) = (\mathcal{T}^n_h \uI{}{h} (\mathcal{T}^m U)_{m\leq n},V_h).
    \label{eq:err.space.P1}
  \end{equation}
  Therefore, we have
  \begin{equation}
    \begin{aligned}
      \mathcal{L}_h&((\ID - P_{\mathfrak{H}})\left(\uI{}{h}\mathcal{T}^nU - \mathcal{T}^n_h (\hat{U}_h^m)_{m\leq n}\right), V_h)\\
      \overset{\eqref{eq:err.space.P1}}&= (J\DIFF_{h} \uI{}{h} \mathcal{T}^n U,V_h)_h - (\uI{}{h} \mathcal{T}^nU,J\DIFF_hV_h)_h
      - (\mathcal{T}^n_h \uI{}{h} (\mathcal{T}^m U)_{m\leq n},V_h)_h\\
      \overset{\eqref{eq:def.esph}}&= (J\DIFF_{h} \uI{}{h} \mathcal{T}^n U,V_h)_h - (\uI{}{h}\DIFF^\star J \mathcal{T}^nU,V_h)_h
    - \epsilon_h(J\mathcal{T}^nU,V_h)
      - (\mathcal{T}^n_h \uI{}{h} (\mathcal{T}^m U)_{m\leq n},V_h)_h \\
      \overset{\ref{asm:disc.spaces.cochain},\ref{asm:disc.time.cochain}}&= 
      (\uI{}{h}J \DIFF \mathcal{T}^n U,V_h)_h + (\uI{}{h}J \DIFF^\star \mathcal{T}^nU,V_h)_h
    - \epsilon_h(J\mathcal{T}^nU,V_h)
    - (\uI{}{h} \mathcal{T}^n (\mathcal{T}^m U)_{m\leq n},V_h)_h \\
    \overset{\ref{asm:disc.time.com},\eqref{eq:def.espT}}&=
    (\uI{}{h}\mathcal{T}^n \cancel{\left(J(\DIFF+\DIFF^\star)U - \partial_t U\right)},V_h)_h
    - \epsilon_h(J\mathcal{T}^nU,V_h)
    - (\uI{}{h} \mathcal{T}^n (\epsilon_{\Delta t,m}(U))_{m\leq n},V_h)_h \\
    &\leq \left( \opn{1,h}{\epsilon_h(\mathcal{T}^nU,\cdot)} + \norm{h}{\uI{}{h} \mathcal{T}^n (\epsilon_{\Delta t,m}(U))_{m\leq n}}\right)\norm{1,h}{V_h},
    \end{aligned}
    \label{eq:err.space.P2}
  \end{equation}
  where we removed the projector $P_{\mathfrak{H}}$ using the fact that $\DIFF P_{\mathfrak{H}} = 0$
  and $\DIFF^* P_{\mathfrak{H}} = 0$, 
  the cancellation occurs because $U$ is solution of \eqref{eq:HDADM.var}, 
  and we used a Cauchy-Schwartz inequality together with $\norm{h}{V_h} \leq \norm{1,h}{V_h}$ on the last line.
  We conclude evaluating \eqref{eq:err.space.P2} for $V_h$ realizing the inf-sup inequality of Lemma \ref{lem:infsup.Lh}.
\end{proof}

\begin{lemma}
  \label{lem:err.exact.space}
  Let $\hat{U}_h^N$ be solution of \eqref{eq:def.hatUh}, 
  and $C_L$ be the constant given by Lemma \ref{lem:infsup.Lh}.
  It holds
  \[
    \norm{1,h}{(\ID-P_{\mathfrak{H}})\left(\uI{}{h}U(T) - \hat{U}_h^N\right)} \leq \frac{1}{C_L} \left( \opn{1,h}{\epsilon_h(U,\cdot)} + \norm{h}{\uI{}{h}\epsilon_{\Delta t,N}(U)}\right).
  \]
\end{lemma}
\begin{proof}
  Using the definition of $\hat{U}_h^N$, we have for any $V_h\in\mathfrak{H}^{\perp_h}$:
  \begin{equation}
    \begin{aligned}
      \mathcal{L}_h((\ID-P_{\mathfrak{H}})\left(\uI{}{h}U(T) - \hat{U}_h^N\right),V_h) 
      &= (J\DIFF \uI{}{h} U,V_h)_h - (\uI{}{h} U,J\DIFF_hV_h)_h - (\uI{}{h} \mathcal{T}^N U,V_h)_h \\
      \overset{\ref{asm:disc.spaces.cochain}}&= 
      (\uI{}{h} \cancel{\left( J(\DIFF + \DIFF^\star) U - \partial_t U \right)},V_h)_h 
      - \epsilon_h(JU,V_h) 
      - (\uI{}{h} \epsilon_{\Delta t,N} U ,V_h)_h \\
    &\leq \left( \opn{1,h}{\epsilon_h(U,\cdot)} + \norm{h}{\uI{}{h} \epsilon_{\Delta t,N}(U))_{m\leq n}}\right)\norm{1,h}{V_h},
    \end{aligned}
    \label{eq:err.exact.space.P1}
  \end{equation}
  where the cancellation occurs because $U$ is solution of \eqref{eq:HDADM.var}.
  We conclude using evaluating \eqref{eq:err.exact.space.P1} for $V_h$ realizing the inf-sup inequality of Lemma \ref{lem:infsup.Lh}.
\end{proof}

\begin{proof}[Proof of Theorem \ref{thm:err.estimate}]
  The well-posedness of the discrete problem \eqref{eq:disc.var} 
  readily follows from the coercivity of $a_h^n$ proven in Lemma \ref{lem:disc.infsup}.
  Let us now prove the error estimate:
  Let $U$ be the solution of \eqref{eq:HDADM.var}.
  For any $0< n \leq N$, applying Lemma \ref{lem:err.time} 
  together with the assumption $\norm{h}{\mathcal{T}^n_{h,L}\left( (\hat{U}_h^m)_{m<n} - (U_h^m)_{m<n} \right)} \leq \theta_n \norm{h}{\hat{U}_h^{n-1} - U_h^{n-1}}$
  gives:
  \[
    \norm{h}{\hat{U}_h^{n} - U_h^{n}} \leq 
    \frac{1}{\theta_n}\norm{h}{(\ID-P_{\mathfrak{H}})\left(\mathcal{T}^n_h(\hat{U}_h^m)_{m\leq n} - \uI{}{h}\mathcal{T}^n U\right)} + \norm{h}{\hat{U}_h^{n-1} - U_h^{n-1}}.
  \]
  Recalling that $U_h^0 = \hat{U}_h^0$, we have
  \begin{equation}
    \norm{h}{\hat{U}_h^N - U_h^N} \leq \sum_{n=1}^N \frac{1}{\theta_n}\norm{h}{(\ID-P_{\mathfrak{H}})\left(\mathcal{T}^n_h(\hat{U}_h^m)_{m\leq n} - \uI{}{h}\mathcal{T}^n U\right)}.
    \label{eq:err.est.P1}
  \end{equation}
  Applying Lemma \ref{lem:err.space} to \eqref{eq:err.est.P1}, we find 
  \begin{equation}
    \norm{h}{\hat{U}_h^N - U_h^N} \leq \frac{1}{C_L} \sum_{n=1}^N \frac{1}{\theta_n}\left( \opn{1,h}{\epsilon_h(\mathcal{T}^nU,\cdot)} + \norm{h}{\uI{}{h} \mathcal{T}^n (\epsilon_{\Delta t,m}(U))_{m\leq n}}\right)
    \label{eq:err.est.P2}
  \end{equation}
  Introducing the assumed bound on $\opn{1,h}{\epsilon_h(\mathcal{T}^nU,\cdot)}$ and 
  $\norm{h}{\uI{}{h} \mathcal{T}^n (\epsilon_{\Delta t,m}(U))_{m\leq n}}$ into \eqref{eq:err.est.P2}, 
  we have
  \begin{equation}
    \begin{aligned}
      \norm{h}{\hat{U}_h^N - U_h^N} &\leq \frac{1}{C_L} \sum_{n=1}^N \frac{1}{\theta_n} \left( (\Delta t)^l + h^{r+1} \right) E(\vert U \vert) \\
      \overset{\ref{asm:disc.time.coer}}&\leq \left( (\Delta t)^l + h^{r+1} \right) \frac{C}{C_L} E(\vert U \vert) .
    \end{aligned}
    \label{eq:err.est.P3}
  \end{equation}
  We conclude using the triangle inequality:
  \[
    \norm{h}{\uI{}{h}U(T) - U_h^N} \leq 
    \norm{h}{(\ID-P_{\mathfrak{H}})\left(\uI{}{h}U(T) - \hat{U}^N_h\right)}
    + \norm{h}{P_{\mathfrak{H}}\left(\uI{}{h}U(T) - \hat{U}^N_h\right)}
    + \norm{h}{\hat{U}_h^N - U_h^N},
  \]
  then using Lemma \ref{lem:err.exact.space} to bound the first term, 
  \eqref{eq:def.hatUh} on the second term,
  and \eqref{eq:err.est.P3} to bound the last term.
\end{proof}

\subsection{Examples: a specific scheme}
In this section we consider a specific choice of discretization in order to prove more properties of the scheme.
We introduce the notation $a \lesssim b$, meaning that there is $C > 0$ depending only on the chosen discrete complex 
such that $a \leq C b$.
We consider a backward Euler time stepping with constant time step $\Delta t$,
setting 
\begin{equation}
  \label{eq:def.CN}
  \mathcal{T}^n_h (V_h^m)_{m \leq n} := \frac{V_h^n - V_h^{n-1}}{\Delta t}, 
  \quad \theta_n = \frac{1}{\Delta t}.
\end{equation}
Let $Y \subset D(\DIFF) \cap D(\DIFF^\star)$ denote a subset of $X$ over which $\uI{}{h}$ is continuous for the $L^2$-norm
(i.e. $\forall V \in Y$, $\norm{h}{\uI{}{h}V} \lesssim \norm{L^2(\Omega)}{V}$), 
and let $U$ be the solution of \eqref{eq:HDADM.var}.
\begin{lemma}[Backward Euler time stepping]
  \label{lem:CN.prop}
  If $U\in C^3([0,T],Y)$, then 
  for all $1\leq n \leq N$, $(V_h^n)_{n\leq N}\in(\uH{}{h})^N$, it holds
  \begin{equation*}
    \begin{aligned}
      \norm{h}{\mathcal{T}^n_{h,L}(V_h^m)_{m<n} } &\leq \theta_n \norm{h}{V_h^{n-1} }, \\
      \norm{h}{\uI{}{h} \epsilon_{\Delta t,N}(U)} +
      \norm{h}{\uI{}{h} \mathcal{T}^n (\epsilon_{\Delta t,m}(U))_{m\leq n}} &\lesssim \Delta t \left(\vert U \vert_{C^2([0,T],Y)} + \vert U \vert_{C^3([0,T],Y)}\right).
    \end{aligned}
  \end{equation*}
\end{lemma}
\begin{proof}
  The first bound is trivial from the definition \eqref{eq:def.CN} of $\mathcal{T}^n_{h,L}$
  giving
  \[
    \norm{h}{\mathcal{T}^n_{h,L}(V_h^m)_{m<n} } = 
    \frac{\norm{h}{V_h^{n-1}}}{\Delta t}
    = \theta_n \norm{h}{V_h^{n-1}}.
  \]
  Assuming the $C^3$ regularity in time of $U$, we can write its Taylor expansion at $t_n$ 
  for $n \geq 2$:
  \begin{equation}
    \begin{aligned}
      U(t_{n-1}) &= U(t_n) - \Delta t \partial_t U(t_n) + (\Delta t)^2 \frac12 \partial_t^2 U(t_n) - (\Delta t)^3\frac16 \partial_t^3U(c) \\
      \partial_t^2 U(t_{n-1}) &= \partial_t^2 U(t_n) - \Delta t \partial_t^3 U(c),
    \end{aligned}
    \label{eq:CN.prop.P1}
  \end{equation}
  where $c \in [t_{n-1},t_n]$.
  Injecting \eqref{eq:CN.prop.P1} into the definition of $\epsilon_{\Delta t,n}(U)$ gives:
  \[
    \epsilon_{\Delta t,n}(U)
      = \frac{U(t_n) - U(t_{n-1})}{\Delta t} - \partial_t U(t_n) \\
      = -\Delta t \frac12 \partial_t^2U(t_n) + (\Delta t)^2 \frac16 \partial_t^3 U(c).
  \]
  Using the same formula for $\epsilon_{\Delta t,n-1}(U)$, we find
  \begin{align*}
    \mathcal{T}^n (\epsilon_{\Delta t,m}(U))_{m\leq n}
    = \frac{\epsilon_{\Delta t,n}(U) - \epsilon_{\Delta t,n-1}(U)}{\Delta t}
    &= -\Delta t \frac12 \frac{\partial_t^2 U(t_n) - \partial_t^2 U(t_{n-1})}{\Delta t}
    + \Delta t \frac16 \left( \partial_t^3 U(c') - \partial_t^3U(c'') \right) \\
    &= - \Delta t \frac12 \partial_t^3 U(c)
    + \Delta t \frac16 \left( \partial_t^3 U(c') - \partial_t^3U(c'') \right).
  \end{align*}
  Taking the $L^2$ norm on both side and bounding $\norm{L^2(\Omega)}{\partial_t^2U(c)}$ 
  by $\norm{C^2([0,T],Y)}{U}$, we have
  \[
    \norm{L^2(\Omega)}{\epsilon_{\Delta t,n}(U)} \lesssim \Delta t\norm{C^2([0,T],Y)}{U},
    \quad \norm{L^2(\Omega)}{\mathcal{T}^n (\epsilon_{\Delta t,m}(U))_{m\leq n}} \lesssim  \Delta t\norm{C^3([0,T],Y)}{U}.
  \]
  We infer the result from the continuity of $\uI{}{h}$ on $Y$.
\end{proof}

\begin{remark}
  The result is straightforward to extend to higher-order schemes approximating the time derivative at $t_n$.
  Other schemes such as the Crank-Nicolson time stepping require some slight modification. 
  For instance, since the Crank-Nicolson time stepping approximates the time derivative at $t_{n-\frac12} := \frac{t_n + t_{n-1}}{2}$, to preserve the second order accuracy, 
  we should use 
  $\epsilon_{\Delta t,n-\frac12}(U) := \frac{U(t_n) - U(t_{n-1})}{\Delta t} - \partial_t U(t_{n-\frac12})$.
  Then we can show that the error 
  \[
    \mathcal{T}^n (\epsilon_{\Delta t,m-\frac12}(U))_{m\leq n}
    = \frac{\epsilon_{\Delta t,n-\frac12}(U) - \epsilon_{\Delta t,n-\frac32}(U)}{\Delta t}
  \]
  is second order accurate.
\end{remark}

\section{Numerical results}
\label{sec:num}

\subsection{Discrete complex}
We use a conforming $\DIVDIV$ complex based on tensor product of splines \cite{bonizzoni2025discrete}. 
The main ingredients in its construction are one dimensional finite elements.
Let $\Eh$ and $\Vh$ denote respectively the set of edges and vertices dividing a segment $[0,1]$.
We consider the following spaces:
\begin{itemize}
  \item $S^{1} := \left\{ v \in C^1([0,1]) \st \forall E \in \Eh, v_{E} \in \Poly{3}(E)  \right\}$
  \item $S^{0} := \left\{ v \in C^0([0,1]) \st \forall E \in \Eh, v_{E} \in \Poly{2}(E)  \right\}$
  \item $S^{-1} := \left\{ v \in L^2([0,1]) \st \forall E \in \Eh, v_{E} \in \Poly{1}(E)  \right\}$
\end{itemize}
The associated interpolator are:
\begin{itemize}
  \item $I^{1}$ such that $\forall v \in C^1([0,1]), \forall V \in \Vh,\ I^1(v)(V) = v(V), (I^1(v))'(V) = v'(V)$.
  \item $I^0$ such that $\forall v \in C^0([0,1]), \forall V \in \Vh,\ I^0(v)(V) = v(V)$, 
    $\forall E \in \Eh$, $\int_E I^0(v) = \int_E v$.
  \item $I^{-1}$ such that $\forall v \in L^2([0,1]), \forall E \in \Eh,\ \forall p \in \Poly{1}(E), 
    \int_E p I^{-1}(v) = \int_E p v$.
\end{itemize}

The basis for the discrete complex are the tensor product of these spaces.
We denote by $S^{a,b,c} := S^a \otimes S^b \otimes S^c$. 
The spaces of the discrete complex are:
\begin{equation*}
  \begin{gathered}
    \uH{0}{h} := \begin{pmatrix}
      S^{1,0,0} \\ S^{0,1,0} \\ S^{0,0,1}
    \end{pmatrix}, \quad
    \uH{1}{h} := \begin{pmatrix}
      S^{0,0,0} & S^{1,-1,0} & S^{1,0,-1} \\
      S^{-1,1,0} & S^{0,0,0} & S^{0,1,-1} \\
      S^{-1,0,1} & S^{0,-1,1} & S^{0,0,0}
    \end{pmatrix}, \quad \\
    \uH{2}{h} := \begin{pmatrix}
      S^{1,-1,-1} & S^{0,0,-1} & S^{0,-1,0} \\
      S^{0,0,-1} & S^{-1,1,-1} & S^{-1,0,0} \\
      S^{0,-1,0} & S^{-1,0,0} & S^{-1,-1,1}
    \end{pmatrix}, \quad
    \uH{3}{h} := S^{-1,-1,-1}.
  \end{gathered}
\end{equation*}
All the spaces are conforming, hence we can take the restriction of the continuous differential as the discrete differential \cite{bonizzoni2025discrete}. 

\subsection{Wave-like solutions}
In order to numerically validate our scheme, 
we consider two classes of wave-like solutions. 
The first depends on three parameters $\lambda_1,\lambda_2,c \in \Real$.
We define
\begin{equation}
  \bvec{k} := \begin{pmatrix}
    0 \\ 0 \\ c
  \end{pmatrix}, \quad
  \bvec{A} := \begin{pmatrix}
    0 & \lambda_1 & 0 \\
    -\lambda_2 & 0 & 0 \\
    0 & 0 & 0
  \end{pmatrix}e^{i(\bvec{k}\cdot\bvec{x} - ct)}, \quad
  \bvec{\gamma} := \begin{pmatrix}
    \frac{\lambda_1-\lambda_2}{2} & 0 & 0 \\
    0 & \frac{\lambda_2 - \lambda_1}{2} & 0 \\
    0 & 0 & -\frac{\lambda_1 + \lambda_2}{2}
  \end{pmatrix}e^{i(\bvec{k}\cdot\bvec{x} - ct)}.
  \label{eq:def.sol.1}
\end{equation}
The second class depends on two parameters $\lambda, c \in\Real$.
We define
\begin{equation}
  \bvec{k} := \begin{pmatrix}
    c \\ c \\ 0
  \end{pmatrix}, \quad
  \bvec{A} := \begin{pmatrix}
    0 & 0 & -\lambda \\
    0 & 0 & -\lambda \\
    0 & 0 & 0
  \end{pmatrix}e^{i(\bvec{k}\cdot\bvec{x} - ct)}, \quad
  \bvec{\gamma} := \begin{pmatrix}
    \lambda & 0 & 0 \\
    0 & -\lambda & 0 \\
    0 & 0 & 0
  \end{pmatrix}e^{i(\bvec{k}\cdot\bvec{x} - ct)}.
  \label{eq:def.sol.2}
\end{equation}
Notice that in the second case $\bvec{k}\cdot\bvec{k} = 2c^2$.

We readily verify that both \eqref{eq:def.sol.1} and \eqref{eq:def.sol.2}
satisfy \eqref{eq:HDADM.HDW}.

\subsection{Results}
We have implemented two test cases: 
The first one is given by \eqref{eq:def.sol.1} with $c := \pi$, $\lambda_1 := 2$, $\lambda_2 := 1$, 
and the second one is given by \eqref{eq:def.sol.2} with $c := \pi$ and $\lambda := 2$.
In either case, we initialized the solution with the reference one at $t = 0$, and let the system evolve until $t = 2\pi$. 
The domain consists of a unit cube, and the mesh is a Cartesian grid. 
We enforced a Dirichlet boundary condition on the whole boundary, 
deriving the value from the analytical solution.
With the setting, the only harmonic forms are the $3$-forms spanned by the linear polynomials. 
  Since the component in the space of $3$-forms of our analytical solutions corresponds to  
$\lambda_3$, and is taken to be zero,
the $\norm{h}{(P_{\mathfrak{H}}\uI{}{h} - \uI{}{h}P_{\mathfrak{H}})U}$ term vanishes in 
Theorem \ref{thm:err.estimate}.
We explored various time steps and spatial subdivisions, 
ranging between $\Delta t = 10^{-2}$ and $\Delta t = 10^{-5}$ for the time steps, 
and between $N = 2$ to $N = 10$ subdivision (hence between $8$ and $1000$ cells).
In order to save on resources, we did not run the simulation for all possible combinations.

The error computed for the various fields with respect to the cell size $h$ using the backward Euler time stepping 
is given in Figure \ref{fig:errorcurves}. 
The value showed for the error is $L^1([0,2\pi]) \times L^2([0,1]^3)$ norm of the local error,
\[
  E := \int_{t = 0}^{2\pi} \left( \int_{x \in [0,1]^3} \norm{}{u_h - \uI{}{h}(u)}^2 \right)^{\frac12} .
\]
We notice that the error attributed to the spatial discretization converges quickly below the contribution of the temporal discretization.
The error on $\lambda_0$ is initially of the order of machine precision and grows due to accumulation errors 
(between $10^{-13}$ and $10^{-7}$).
The evolution of the error during the simulation is given in Figure \ref{fig:errorstime} for a 
case dominated by the time discretization, 
one dominated by the spatial discretization, 
and one intermediary.

Since we could not neglect the error due to the time discretization, 
we modeled our error as $E = \gamma\left( (\Delta t)^\alpha + \delta h^\beta \right)$,
and computed the convergence rates $\alpha$ and $\beta$ as a best fit for the data obtained in our simulation.
The results obtained are shown in Table \ref{tab:convrate}.
The results are consistent with a first order time discretization.
Since there exists bounded cochain projections to the spline complex \cite{bonizzoni2025discrete}, 
we expect to see a spatial convergence of order $k+1$ where $k$ is the polynomial degree.
The discrete spaces for the space of $1$, $2$, and $3$-forms contain all polynomial of degree $1$, 
but not higher in a certain direction. 
Hence, the standard convergence theory gives a convergence rate of $2$.
The computed convergence rate is compatible with the result, even significantly higher.
  Two possible explanations for this higher convergence rate are 
  the alignment between our analytical solutions, and the directions containing higher order polynomials in the discrete spaces, 
  or a superconvergence due to the use of a smooth solution, and the use of the discrete norm to compute the error.

  To confirm the convergence rates, we also implemented the Crank-Nicolson time stepping that is second order accurate in time. We compute the error using the same formula on the same test cases.
  The results are given in Figure \ref{fig:errorcurvesCN}.
  Since the contributions to the error from the time discretization are neglectable, 
  we computed the convergence rate in space only using a timestep of $\Delta t = 10^{-4}$.
  The results are shown in Table \ref{tab:convrateCN}.
  They are coherent with the value obtained using the backward Euler time stepping.

\begin{table}[]
  \centering
\begin{tabular}{l|cc|cc|}
\cline{2-5}
                                  & \multicolumn{1}{l}{First case} & \multicolumn{1}{l|}{} & \multicolumn{1}{l}{Second case} & \multicolumn{1}{l|}{} \\ \cline{2-5} 
                                  & \multicolumn{1}{c|}{$\alpha$}  & $\beta$               & \multicolumn{1}{c|}{$\alpha$}   & $\beta$               \\ \hline
\multicolumn{1}{|l|}{g}           & \multicolumn{1}{c|}{0.986}     & 3.31                  & \multicolumn{1}{c|}{0.987}      & 3.20                  \\ \hline
\multicolumn{1}{|l|}{A}           & \multicolumn{1}{c|}{0.990}     & 3.13                  & \multicolumn{1}{c|}{0.996}      & 3.05                  \\ \hline
\multicolumn{1}{|l|}{$\lambda_3$} & \multicolumn{1}{c|}{1.01}      & 4.40                  & \multicolumn{1}{c|}{1.06}       & 4.11                  \\ \hline
\end{tabular}
\caption{Computed convergence rates using the Euler time stepping.}
\label{tab:convrate}
\end{table}

\begin{table}[]
  \centering
\begin{tabular}{l|c|c|}
  \cline{2-3}
                                  & {First case} & {Second case} \\ \hline 
\multicolumn{1}{|l|}{g}           &  3.22                  &  3.12                  \\ \hline
\multicolumn{1}{|l|}{A}           &  3.09                  &  3.03                  \\ \hline
\multicolumn{1}{|l|}{$\lambda_3$} &  4.16                  &  4.07                  \\ \hline
\end{tabular}
\caption{Computed convergence rates (in space) using the Crank-Nicolson time stepping.}
\label{tab:convrateCN}
\end{table}

\begin{figure}[p] 
    \centering
    \begin{subfigure}[t]{0.48\textwidth}
        \centering
        \includegraphics[width=\textwidth]{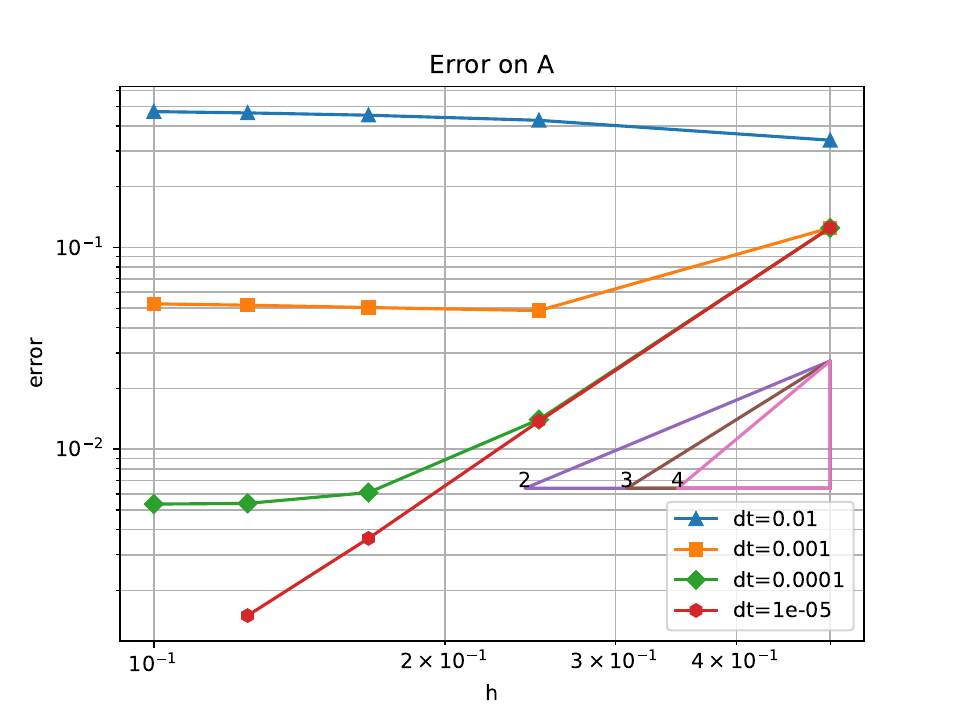}
        \caption{Error on $A$ in the first case.}
    \end{subfigure}
    \hfill
    \begin{subfigure}[t]{0.48\textwidth}
        \centering
        \includegraphics[width=\textwidth]{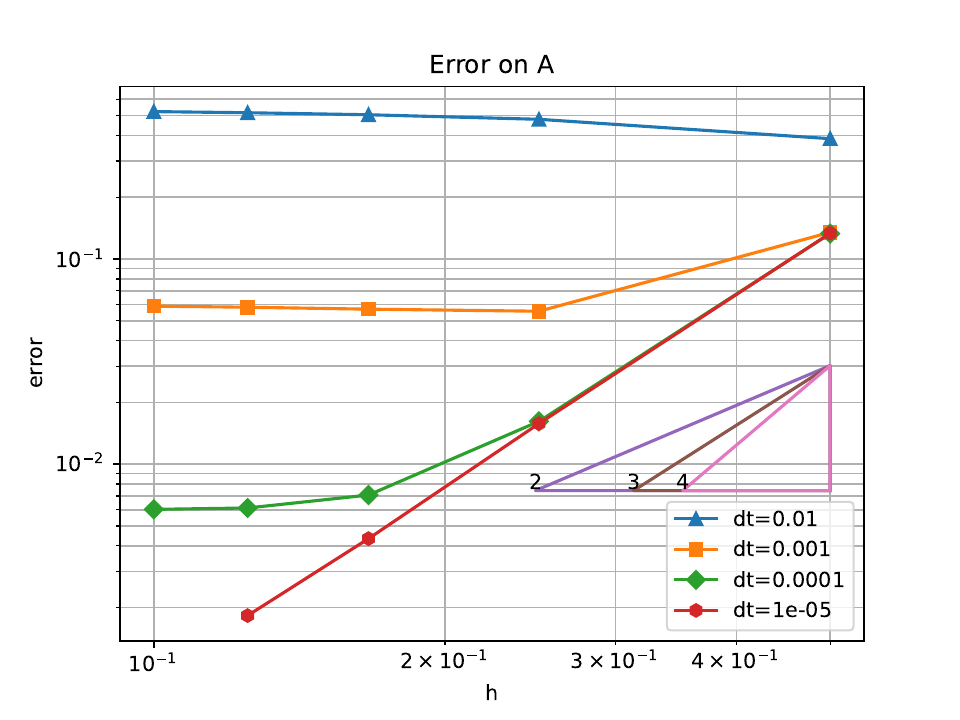}
        \caption{Error on $A$ in the second case.}
    \end{subfigure}
    \vspace{0.5cm}
    \begin{subfigure}[t]{0.48\textwidth}
        \centering
        \includegraphics[width=\textwidth]{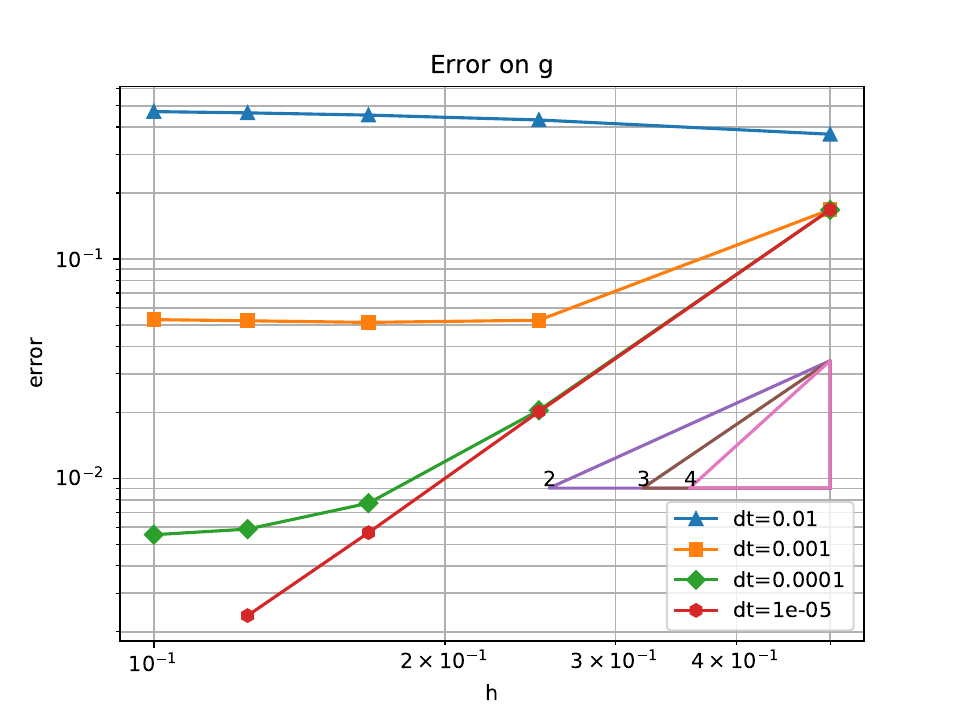}
        \caption{Error on $g$ in the first case.}
    \end{subfigure}
    \hfill
    \begin{subfigure}[t]{0.48\textwidth}
        \centering
        \includegraphics[width=\textwidth]{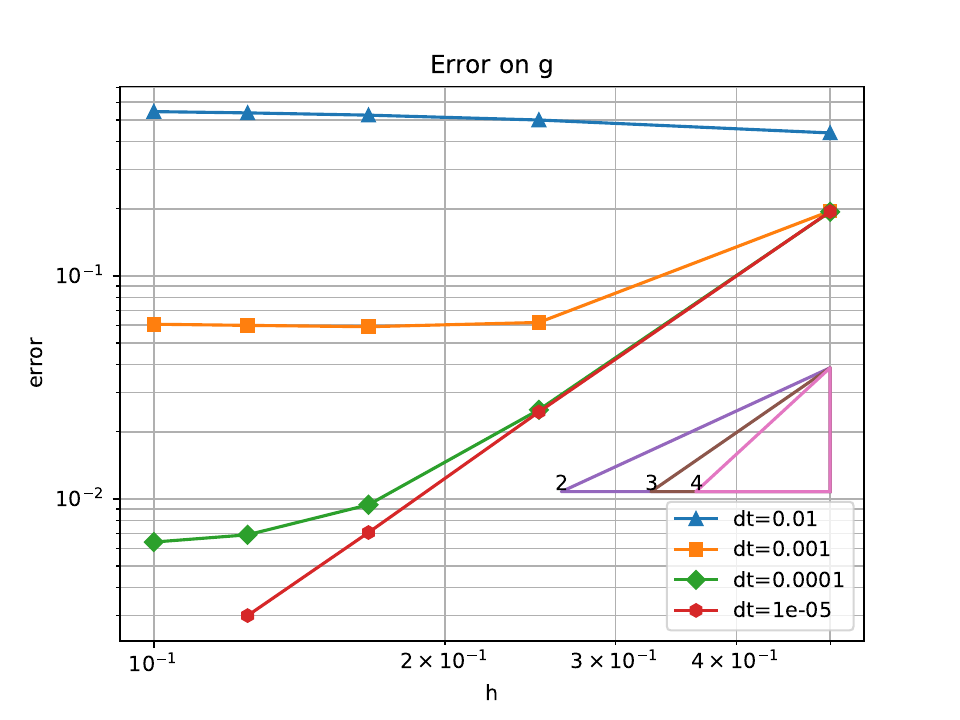}
        \caption{Error on $g$ in the second case.}
    \end{subfigure}
    \vspace{0.5cm}
    \begin{subfigure}[t]{0.48\textwidth}
        \centering
        \includegraphics[width=\textwidth]{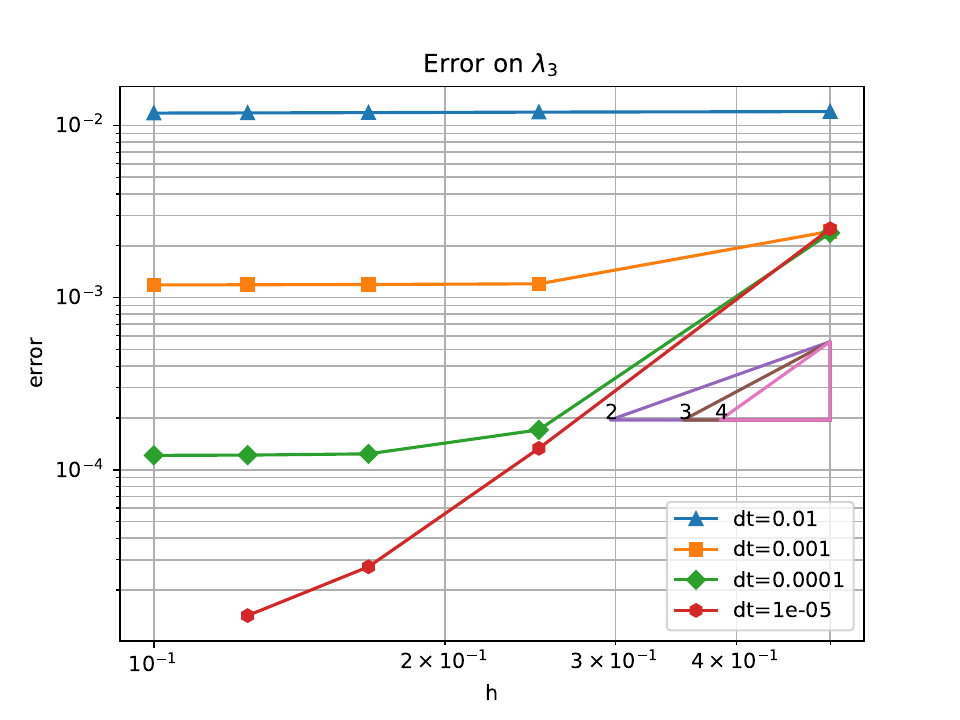}
        \caption{Norm of $\lambda_3$ in the first case.}
    \end{subfigure}
    \hfill
    \begin{subfigure}[t]{0.48\textwidth}
        \centering
        \includegraphics[width=\textwidth]{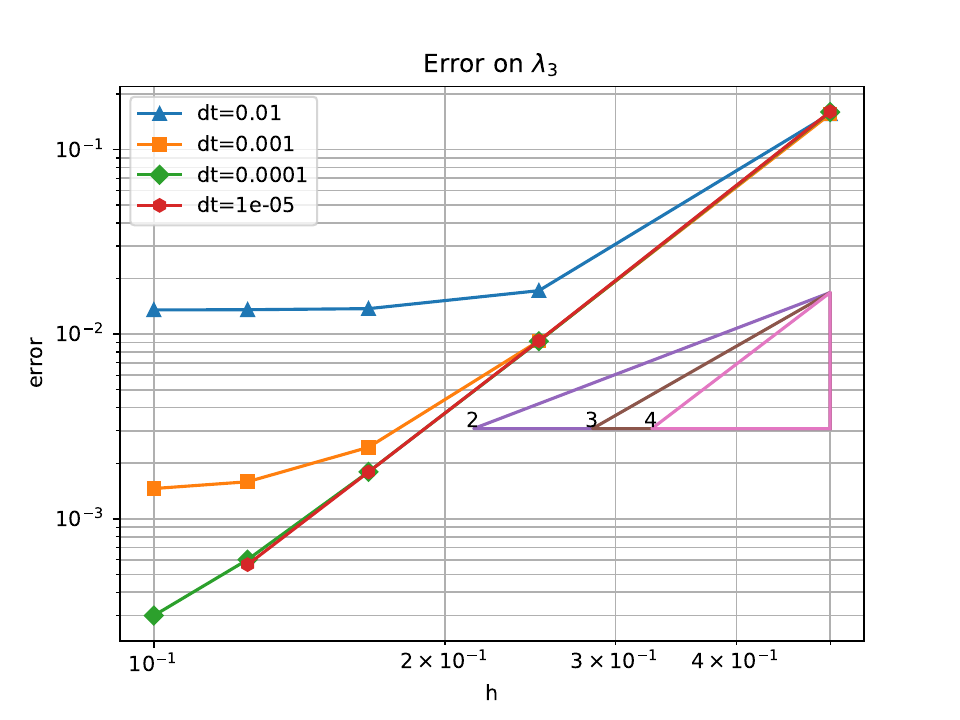}
        \caption{Norm of $\lambda_3$ in the second case.}
    \end{subfigure}
    \caption{Error with respect to the mesh size for various time steps using the Euler time stepping.}
    \label{fig:errorcurves}
\end{figure}

\begin{figure}[p] 
    \centering
    \begin{subfigure}[t]{0.48\textwidth}
        \centering
        \includegraphics[width=\textwidth]{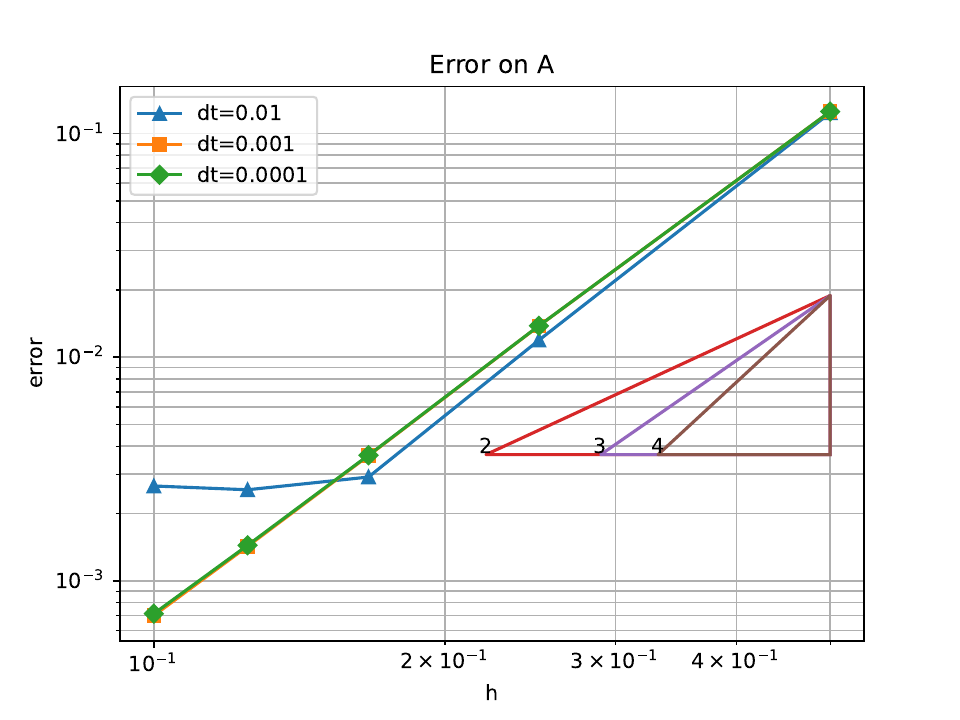}
        \caption{Error on $A$ in the first case.}
    \end{subfigure}
    \hfill
    \begin{subfigure}[t]{0.48\textwidth}
        \centering
        \includegraphics[width=\textwidth]{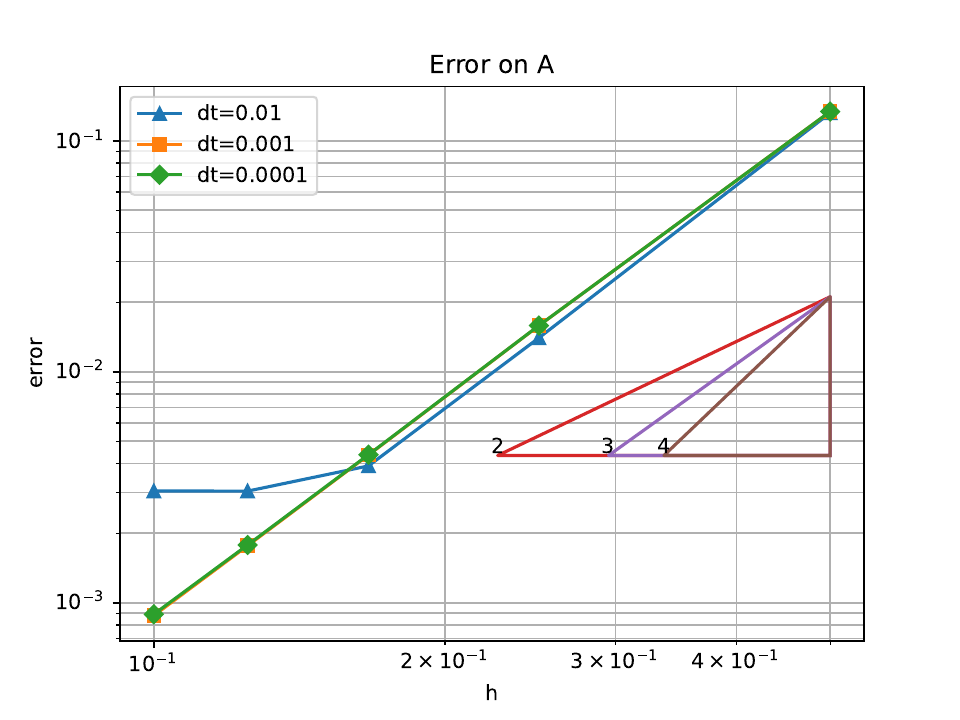}
        \caption{Error on $A$ in the second case.}
    \end{subfigure}
    \vspace{0.5cm}
    \begin{subfigure}[t]{0.48\textwidth}
        \centering
        \includegraphics[width=\textwidth]{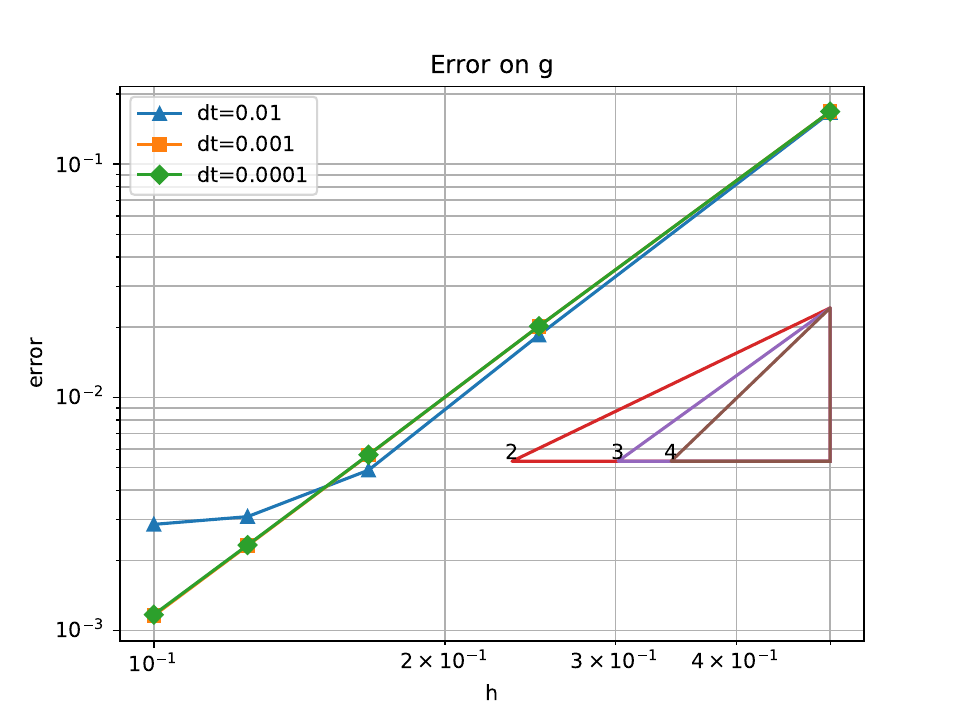}
        \caption{Error on $g$ in the first case.}
    \end{subfigure}
    \hfill
    \begin{subfigure}[t]{0.48\textwidth}
        \centering
        \includegraphics[width=\textwidth]{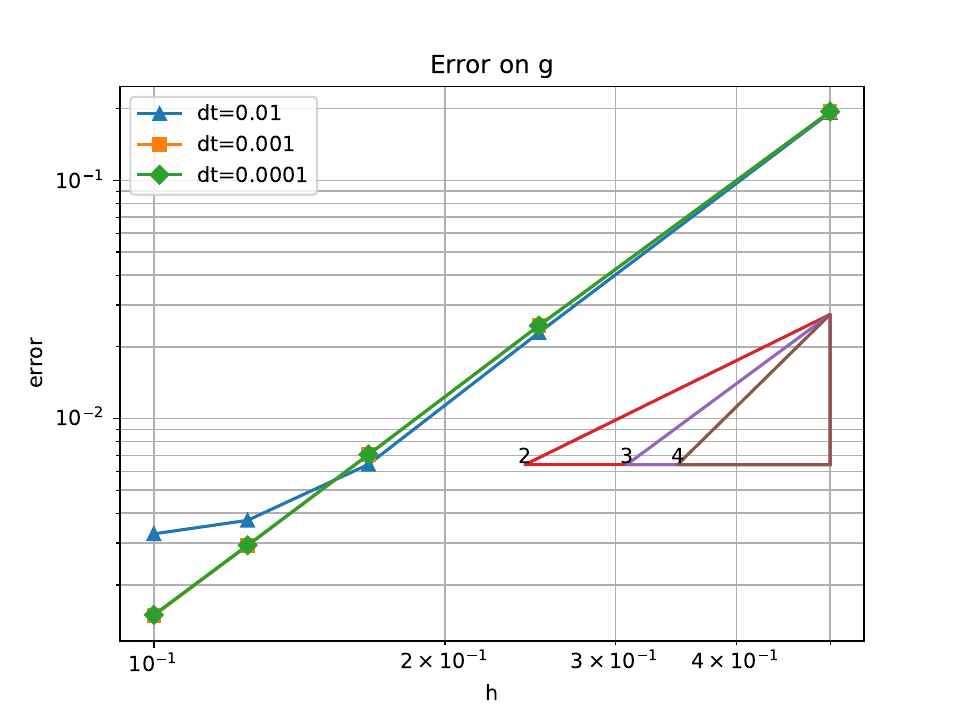}
        \caption{Error on $g$ in the second case.}
    \end{subfigure}
    \vspace{0.5cm}
    \begin{subfigure}[t]{0.48\textwidth}
        \centering
        \includegraphics[width=\textwidth]{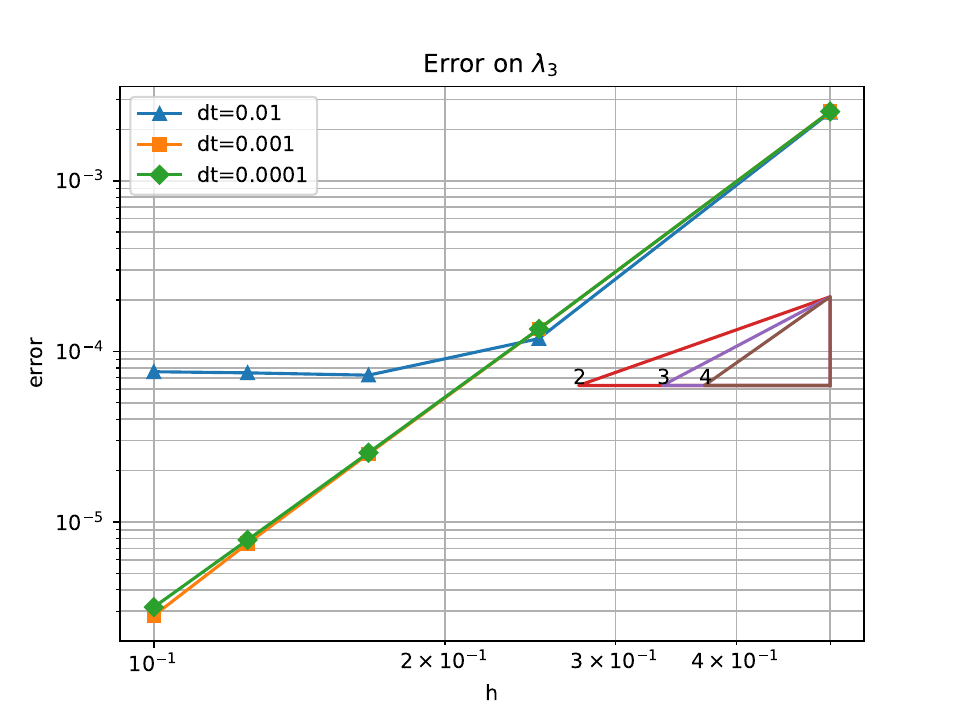}
        \caption{Norm of $\lambda_3$ in the first case.}
    \end{subfigure}
    \hfill
    \begin{subfigure}[t]{0.48\textwidth}
        \centering
        \includegraphics[width=\textwidth]{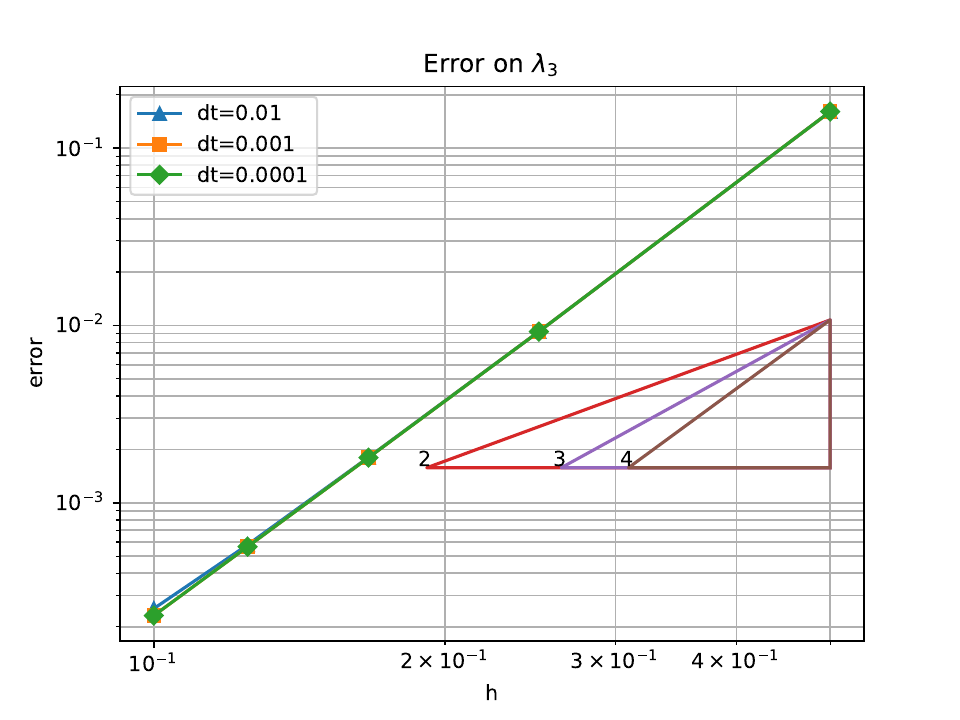}
        \caption{Norm of $\lambda_3$ in the second case.}
    \end{subfigure}
    \caption{Error with respect to the mesh size for various time steps using the Crank-Nicolson time stepping.}
    \label{fig:errorcurvesCN}
\end{figure}

\begin{figure}[p] 
    \centering
    \begin{subfigure}[t]{\textwidth}
        \centering
        \includegraphics[width=0.6\textwidth]{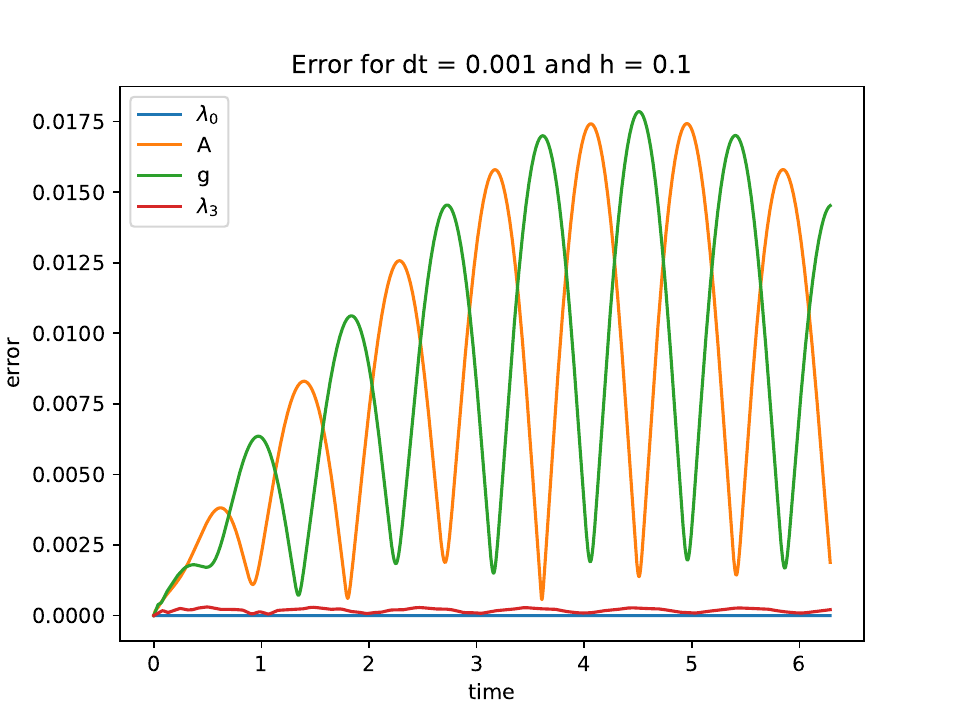}
        \caption{Error dominated by the time approximation.}
    \end{subfigure}
    \vspace{0.5cm}
    \begin{subfigure}[t]{\textwidth}
        \centering
        \includegraphics[width=0.6\textwidth]{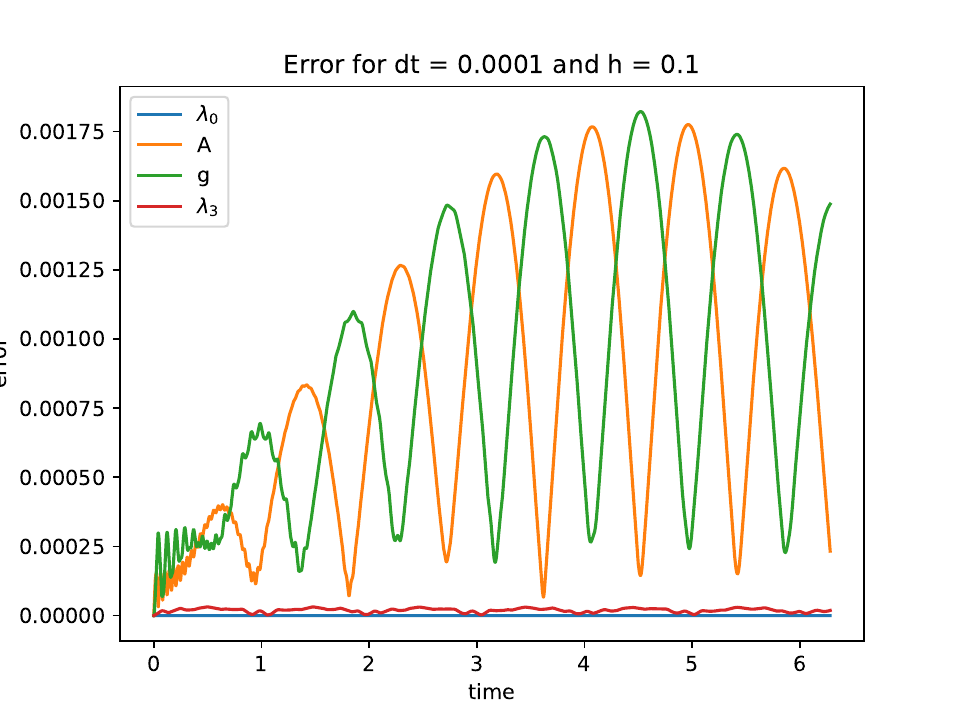}
        \caption{Error with comparable contribution from time and space.}
    \end{subfigure}
    \vspace{0.5cm}
    \begin{subfigure}[t]{\textwidth}
        \centering
        \includegraphics[width=0.6\textwidth]{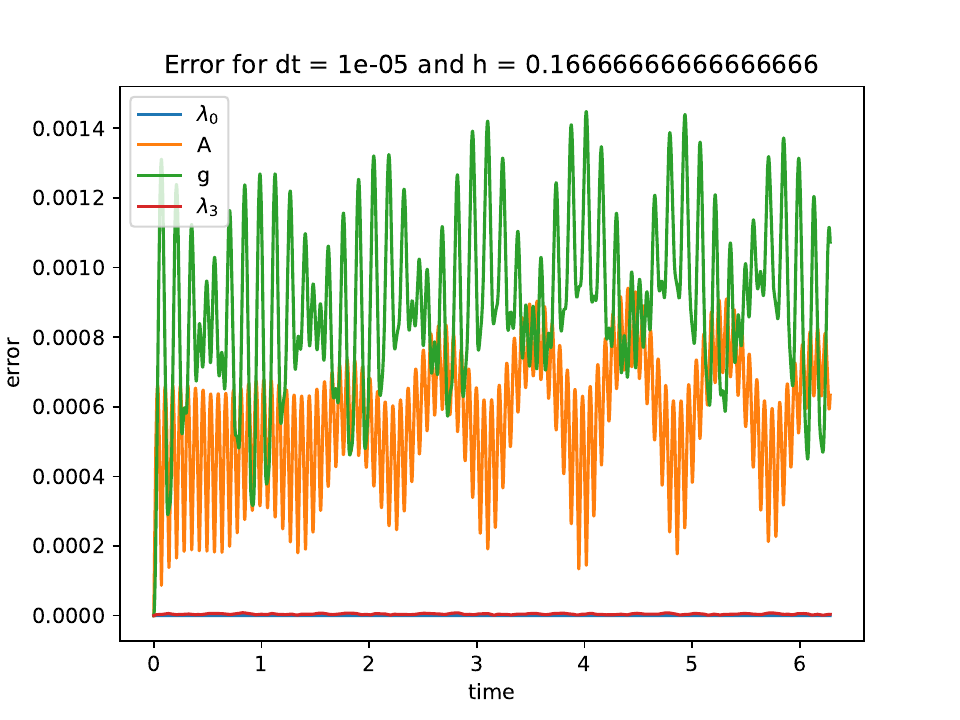}
        \caption{Error dominated by the spatial approximation.}
    \end{subfigure}
    \caption{Evolution of the error during the simulated time for the first case.}
    \label{fig:errorstime}
\end{figure}

\section{Conclusions}

We aim to extend this stability to the full, nonlinear equations. The linearized version provides a promising first step, since the mathematical properties of nonlinear hyperbolic systems mainly rely on their principal parts. However, deriving the associated nonlinear mixed formulation is   beyond the scope of this paper and is left as a future direction.

\appendix
\section{Vector calculus identities}
\begin{lemma}
  If $\gamma \in \Symm$ then
  \begin{equation}
    \SYM\CURL\CURL S\gamma = \INC \gamma - S\DEF\VDIV S\gamma - S\HESS\tr\gamma.
    \label{eq:symcurlcurl.dec}
  \end{equation}
  \label{lem:symcurlcurl.dec}
\end{lemma}
\begin{proof}
  The proof relies on the identities 
  \begin{equation}
    \begin{aligned}
      \CURL\iota &= \mskw\GRAD \\
      2\skw\CURL &= \mskw2\vskw\CURL = -\mskw\VDIV S,
    \end{aligned}
    \label{eq:pr.ids.C1}
  \end{equation}
  and 
  \begin{equation}
    \begin{aligned}
      \CURL\mskw &= S\GRAD \\
      \SYM\CURL\mskw &= S\DEF.
    \end{aligned}
    \label{eq:pr.ids.C2}
  \end{equation}

  If $\gamma \in \Symm$, then $\CURL\gamma\in\Tless$ and $S^{-1}\CURL\gamma = T\CURL\gamma = \CURL\gamma - 2 \skw\CURL\gamma$. 
  Expanding the definition of $S$, we have
  \begin{equation*}
    \begin{aligned}
      \SYM\CURL\CURL S\gamma 
      &= \SYM\CURL\left(\CURL\gamma - \CURL\iota\tr\gamma\right) \\
    &= \SYM\CURL\left(S^{-1}\CURL\gamma + 2\skw\CURL\gamma - \CURL\iota\tr\gamma\right) \\
    \overset{\eqref{eq:pr.ids.C1}}&= \SYM\CURL\left(S^{-1}\CURL\gamma - \mskw\VDIV S\gamma - \mskw\GRAD\tr\gamma\right) \\
    \overset{\eqref{eq:pr.ids.C2}}&= \SYM\INC\gamma - S\DEF\VDIV S\gamma - \SYM S\HESS\tr\gamma \\
      &= \INC\gamma - S\DEF\VDIV S\gamma - S\HESS\tr\gamma.
    \end{aligned}
  \end{equation*}
\end{proof}

\begin{lemma}
  If $\VDIV A = \bvec{0}$ then
  \begin{equation}
    \VDIV\SYM\CURL A = -\CURL\CURL\vskw A .
    \label{eq:divsymcurl.dec}
  \end{equation}
  \label{lem:divsymcurl.dec}
\end{lemma}
\begin{proof}
  The proof relies on the identity 
  \begin{equation}
    \VDIV\mskw = \CURL.
    \label{eq:pr.ids.C3}
  \end{equation}
  Expanding the definition of $\SYM$, we have
  \begin{equation*}
    \begin{aligned}
      \VDIV\SYM\CURL A
      &= \frac12 \cancel{\VDIV\CURL A} + \frac12 \VDIV T \CURL A \\
      \overset{\eqref{eq:divTcurl.curldivT}} &= \frac12\CURL\VDIV A^\top \\
      &= \frac12\CURL\cancel{\VDIV A} - \CURL\VDIV\skw A \\
      &= -\CURL\VDIV \mskw\vskw A \\
      \overset{\eqref{eq:pr.ids.C3}}&= -\CURL\CURL\vskw A,
    \end{aligned}
  \end{equation*}
  where we used the decomposition $A^\top = A - 2\skw A$ on the third line.
\end{proof}

\begin{lemma}
  In general, it holds
  \begin{align}
    \CURL\CURL &= -2\VDIV S\DEF \label{eq:curlcurl.divdef},\\
    \VDIV T \CURL &= \CURL \VDIV T \label{eq:divTcurl.curldivT},
  \end{align}
  where $T$ is the transpose operator.
  \label{lem:curl.comms}
\end{lemma}
\begin{proof}
  The result is a straightforward computation. \\
  \ul{Proof of \eqref{eq:curlcurl.divdef}}:
  Using Einstein notations, for any vector field $v$, we have
  $(\CURL v)_i = \epsilon_{ijk}\partial_jv_k$ where $\epsilon$ is the fully skew-symmetric tensor.
  Therefore, 
  \begin{equation}
    (\CURL\CURL v)_i 
    = \epsilon_{ijk}\partial_j(\CURL v)_k 
    = \epsilon_{ijk} \epsilon_{klm} \partial_j\partial_l v_m 
    = (\delta_i^l\delta_j^m - \delta_i^m\delta_j^l)\partial_j\partial_lv_m
    = \partial_j\partial_iv_j - \partial_j\partial_jv_i, 
    \label{eq:pr.curlcurl.einst}
  \end{equation}
  which is the vector Laplacian identity $-\Delta =\CURL\CURL-\GRAD\DIV$.
  On the other hand, writing $S \DEF v = \DEF v - \iota\tr\GRAD v$, we have
  \[ 
    2(S\DEF v)_{ji} = 2(\DEF v)_{ji} - 2(\iota\tr\GRAD v)_{ji} 
    = \partial_j v_i + \partial_i v_j - 2\delta^i_j \partial_k v_k.
  \]
  We infer \eqref{eq:curlcurl.divdef} equating \eqref{eq:pr.curlcurl.einst} with 
  \[
    -2(\VDIV S\DEF v)_i 
    = -\partial_j (2 S\DEF v)_{ji}
    = 2 \partial_i\partial_jv_j - \partial_j\partial_iv_j - \partial_j\partial_jv_i 
    = \partial_i\partial_jv_j - \partial_j\partial_jv_i .
  \]
  \ul{Proof of \eqref{eq:divTcurl.curldivT}}:
  For any matrix field $M$, we have by convention 
  $(\CURL M)_{ij} = \epsilon_{ikl}\partial_k M_{lj}$.
  Thus 
  \begin{equation}
    (\VDIV T \CURL M)_i = \partial_j (T\CURL M)_{ji} = \partial_j (\CURL M)_{ij} = \epsilon_{ikl}\partial_j \partial_k M_{lj}.
    \label{eq:pr.divTcurl.einst}
  \end{equation}
  On the other hand, we have $(\VDIV M^\top)_i = \partial_j (M^\top)_{ji} = \partial_j M_{ij}$.
  We infer \eqref{eq:divTcurl.curldivT} equating \eqref{eq:pr.divTcurl.einst} with
  \[
    (\CURL\VDIV M^\top)_i = \epsilon_{ikl} \partial_k (\VDIV M^\top)_l
    = \epsilon_{ikl}\partial_k\partial_j M_{lj} .
  \]
\end{proof}

\section*{Acknowledgement}

The work was supported by a Royal Society University Research Fellowship (URF$\backslash$R1$\backslash$221398, RF$\backslash$ERE$\backslash$221047), an ERC Starting Grant (project 101164551, GeoFEM) and a Royal Society International Exchanges Grant (IEC$\backslash$NSFC$\backslash$233594).  Views and opinions expressed are however those of the authors only and do not necessarily reflect those of the European Union or the European Research Council. Neither the European Union nor the granting authority can be held responsible for them.

\printbibliography
\end{document}